\documentclass[a4paper,abstract=true]{scrartcl}
\usepackage[margin=1in]{geometry}
\usepackage{amsfonts}
\usepackage{amssymb}
\usepackage{amsmath}
\usepackage{amsthm}
\usepackage{comment}
\usepackage{graphicx}
\usepackage{subcaption}
\usepackage{thm-restate}

\usepackage[ruled,vlined,linesnumbered]{algorithm2e}
\usepackage{hyperref}
\usepackage[noabbrev,capitalize]{cleveref}

\usepackage{csquotes}
\usepackage{lineno} 
\usepackage{tikz}
\usetikzlibrary{calc}
\usetikzlibrary{patterns,decorations.pathreplacing}
\usetikzlibrary{arrows.meta}
\newtheorem{theorem}{Theorem}

\newtheorem{lemma}[theorem]{Lemma}

\newtheorem{cor}[theorem]{Corollary}
\newtheorem{definition}[theorem]{Definition}

\newtheorem{question}{Question}
\newtheorem*{definition*}{Definition}
\newtheorem*{lemma*}{Lemma}

\usepackage{authblk}

\allowdisplaybreaks

\newcommand{\R}{\mathbb{R}}
\newcommand{\N}{\mathbb{N}}

\newcommand{\Q}{\mathbb{Q}}

\newcommand{\set}[1]{\{ #1 \}}
\newcommand{\fromto}[2]{\set{#1,\dots,#2}}

\DeclareMathOperator{\diam}{diam}
\DeclareMathOperator{\dist}{dist}
\DeclareMathOperator{\cdiam}{cdiam}

\newcommand{\symdiff}{\mathbin{\triangle}}
\newcommand{\vout}{\text{out}}
\newcommand{\vin}{\text{in}}

\newcommand{\La}{\mathcal{L}}
\newcommand{\C}{\mathcal{C}}

\begin{document}

\title{Computing the Polytope Diameter is Even Harder than NP-hard (Already for Perfect Matchings)}

\author{Lasse Wulf}
%\author{Anonymous Author(s)}

\affil{IT University of Copenhagen, Copenhagen, Denmark, \texttt{lasw@itu.dk}}
%\affil{Anonymous affiliations}

\date{}

\maketitle

\begin{abstract}
The diameter of a polytope is a fundamental geometric parameter that plays a crucial role in understanding the efficiency of the simplex method. 
Despite its central nature, the computational complexity of computing the diameter of a given polytope is poorly understood.
Already in 1994, Frieze and Teng [Comp. Compl.] recognized the possibility that this task could potentially be harder than NP-hard, and asked whether the corresponding decision problem is complete for the second level of the polynomial hierarchy, i.e.\ $\Pi^p_2$-complete. In the following years, partial results could be obtained. In a cornerstone result, Frieze and Teng themselves proved weak NP-hardness for a family of custom defined polytopes. Sanità [FOCS18] in a break-through result proved that already for the much simpler fractional matching polytope the problem is strongly NP-hard. Very recently, Steiner and Nöbel [SODA25] generalized this result to the even simpler bipartite perfect matching polytope and the circuit diameter.
In this paper, we finally show that computing the diameter of the bipartite perfect matching polytope is $\Pi^p_2$-hard. 
Since the corresponding decision problem is also trivially contained in $\Pi^p_2$, this decidedly answers Frieze and Teng's 30 year old question.

Our results in particular hold even when the constraint matrix of the given polytope is totally unimodular. They also hold when the diameter is replaced by the circuit diameter.

As our second main result, we prove that for some $\varepsilon > 0$ the (circuit) diameter of the bipartite perfect matching polytope cannot be approximated by a factor better than $(1 + \varepsilon)$.
This answers a recent question by Nöbel and Steiner. It is the first known inapproximability result for the circuit diameter, and extends Sanità's inapproximability result of the diameter to the totally unimodular case.

\end{abstract}


%\vspace*{0.4cm}

\noindent\textbf{Funding:} This work was supported by Eva Rotenberg's Carlsberg Foundation Young Researcher Fellowship CF21-0302 ``Graph Algorithms with Geometric Applications''.

%\noindent\textbf{Acknowledgments:} Anonymous acknowledgments

%\vspace*{0.4cm}

\noindent\textbf{Acknowledgements:} I am thankful to Eva Rotenberg for great encouragement and helpful discussions, and to Hung P.\ Hoang for introducing me to this intriguing problem.

\newpage
%%%%%%%%%%%%%%%%%%%%%%%%%%%%%%%%%%%%%%%%%%%%%%%%%%%%%%%%%%%%%%%%%%%%%%%%%
%%%%%%%%%%%%%%%%%%%%%%%%%%%%%%%%%%%%%%%%%%%%%%%%%%%%%%%%%%%%%%%%%%%%%%%%%

\section{Introduction}
\label{sec:introduction}

The \emph{diameter} of a polytope is the maximum length of a shortest path (in terms of edges) between any two vertices of the polytope. 
This property plays a crucial role in understanding the complexity of algorithms that traverse polytopes, such as the simplex algorithm for linear programming.

The famous simplex algorithm, as invented in its basic form in the 1950s by George Dantzig, aims to solve \emph{linear programming}, i.e.\ to maximize a linear functional $c^tx$ over a polytope $P$. 
%It is possible to understand the workings of the simplex algorithm as follows.
Roughly speaking, the simplex algorithm works such that in each step the algorithm visits one of the vertices of $P$, and a \emph{pivot rule} decides which of the incident edges in the 1-skeleton of $P$ to traverse in order to visit the next vertex. 
The performance of the algorithm is hence closely tied to the polytope's diameter. 
If the diameter is small relative to the polytope's size, fewer steps might be needed to find an optimal solution, making the algorithm more efficient.
However, if the diameter is large relative to the polytope's size, then under \emph{any} pivot rule the simplex algorithm needs to perform many steps in the worst case.
More precisely, if $\diam(P)$ denotes the diameter of $P$, there exists an initialization vertex of the simplex algorithm and a direction $c \in \R^n$, such that the algorithm needs to perform at least $\diam(P)$ steps.

One of the most famous open questions of mathematical optimization, as re-iterated by Smale in his list of unsolved problems for the 21st century \cite{smale1998mathematical}, is whether there exists a strongly polynomial time algorithm for linear programming.
Maybe surprisingly, the simplex algorithm is still a potential candidate for a positive resolution of this conjecture -- while for most known pivot rules counterexamples have been found which 
cause the simplex algorithm to have superpolynomial running time  (see e.g.\ \cite{disser2023exponential,klee1972good}), the existence of a pivot rule such that the simplex algorithm has strongly polynomial runtime still remains a famous open question.

Since the runtime of the simplex algorithm is lower-bounded by $\diam(P)$, a necessary condition for such a strongly polynomial pivot rule would be 
that for every polytope $P \subseteq \R^d$ with $n$ facets, its diameter is bounded by a polynomial function in $n$ and $d$. This condition has become known under the name \emph{polynomial Hirsch conjecture}.
\[
 \diam(P) \leq \text{poly}(n, d) \text{ for all polytopes }P \subseteq \R^d \text{ with $n$ facets.}
\]
Much like the existence of a strongly polynomial pivot rule, the polynomial Hirsch conjecture is a famous conjecture in the theory of mathematical optimization and convex geometry, and still remains wide open. It is a relaxation of the \emph{classic Hirsch conjecture}, which was established by Hirsch in 1957 and stated that $\diam(P) \leq n - d$ for all polytopes.
The classic Hirsch conjecture has since been disproven, first for unbounded polytopes by Klee and Walklup \cite{klee1967d} and finally for bounded polytopes in 2012 by Santos \cite{santos2012counterexample}.
However, even today the best known counterexamples are quite close to Hirsch' original conjecture, in the sense that Matschke, Santos and Weibel found polytopes $P$ with $\diam(P) \geq \frac{21}{20}(n - d)$ \cite{matschke2015width}.
%The best known upper bound on the diameter of a general polytope states that $\diam(P) \leq ...$. It was proven by [] and is still far from polynomial.
In the important special case of 0/1 polytopes, i.e.\ polytopes with the property that all vertices lie on 0/1-integer coordinates, the Hirsch conjecture is known to be true, as shown by Naddef \cite{naddef1989hirsch}.
%Such polytopes are especially important for combinatorial optimization.

\textbf{The complexity of computing the diameter.}
We consider the problem \textsc{Diameter} of computing the diameter of a polytope.
More specifically, since we are interested in questions regarding computational complexity, it is more convenient to talk about its corresponding decision version \textsc{Diameter-Decision}. These problems are defined as follows.
\begin{quote}
    \textbf{Problem}: $\textsc{Diameter}$

    \textbf{Input}: $A \in \Q^{m \times d}$ and $b \in \Q^m$ describing the polytope $P = \set{x \in \R^d : Ax \leq b}$.

    \textbf{Task}: Compute the diameter $\diam(P)$. 
\end{quote}
%--------------------------
\newpage
%--------------------------
\begin{quote}
    \textbf{Problem}: $\textsc{Diameter-Decision}$

    \textbf{Input}: $A \in \Q^{m \times d}$ and $b \in \Q^m$ describing $P = \set{x \in\R^d : Ax \leq b}$, and some $t \in \N$.

    \textbf{Question}: Is $\diam(P) \leq t$? 
\end{quote}

Since the diameter of a polytope is the central object of study of the Hirsch conjecture, which has been established already in 1957, 
one would expect that the computational complexity of $\textsc{Diameter}$ and $\textsc{Diameter-Decision}$ is well understood.
However, that is not the case.
As a cornerstone result, Frieze and Teng \cite{DBLP:journals/cc/FriezeT94} proved that $\textsc{Diameter}$ is weakly NP-hard for a family of custom defined polytopes. 
%Kaibel and Pfetsch \cite{DBLP:conf/dagstuhl/KaibelP03} reiterated the question about how hard it is to compute the diameter. 
Sanità \cite{DBLP:conf/focs/Sanita18} in a break-through result proved that already for the much simpler fractional matching polytope the problem \textsc{Diameter} is strongly NP-hard. 
Very recently, Nöbel and Steiner \cite{nobel2025complexity} generalized this result to the even simpler bipartite perfect matching polytope.
However, none of these papers could establish whether the problem \textsc{Diameter-Decision} is \emph{actually contained} in the class NP.
Indeed, this is not a coincidence. Already in 1994, Frieze and Teng recognized the possibility that the problem \textsc{Diameter-Decision} could be even harder than NP-hard and asked
\begin{question}[Frieze and Teng \cite{DBLP:journals/cc/FriezeT94}]
Is \textsc{Diameter-Decision} $\Pi^p_2$-complete?
\end{question}

Here, the class $\Pi^p_2 = \text{co-}\Sigma^p_2$ denotes the class in the second level of the \emph{polynomial hierarchy} introduced by Stockmeyer \cite{DBLP:journals/tcs/Stockmeyer76}. 
Roughly speaking, the class $\Pi^p_2$ contains all the problems of the form: 
FOR ALL objects $x$, does there EXIST some object $y$ such that some efficiently testable property $\mathcal{P}(x,y)$ holds?
This is a natural fit to \textsc{Diameter-Decision}, since $\diam(P)$ is small if and only if FOR ALL vertex pairs, there EXISTS a short path in $P$ connecting the pair.
Similar to the $\text{P} \neq \text{NP}$ conjecture, complexity theorists believe that $\text{co-NP} \neq \Pi^p_2$, 
or equivalently, that the polynomial hierarchy does not collapse to the first level \cite{DBLP:journals/4or/Woeginger21}.
If this belief is true, it means that every $\Pi^p_2$-complete problem is not contained in NP (or co-NP), since it is complete for a much larger complexity class.
Indeed, we show in this paper that Question~1 has a positive answer, already for the bipartite perfect matching polytope. This provides an explanation why previous papers could not establish membership of \textsc{Diameter-Decision} to NP. 
The precise result together with its implications will be explained below.

\textbf{Bipartite perfect matchings}.
All our results hold already for the very well known and widely studied \emph{bipartite perfect matching polytope}, which has countless important applications in combinatorial optimization.
Given a bipartite graph $G = (V, E)$, 
let us define the characteristic vector $\chi^M \in \R^E$ of some perfect matching $M \subseteq E$ as the edge indexed vector $(x_e)_{e \in E}$ 
that has $x_e = 1$ if $e \in M$ and $x_e = 0$ otherwise. 
The bipartite perfect matching polytope of $G$ is then defined as the convex hull of all characteristic vectors of perfect matchings.
\[
P_G := \text{conv}\set{\chi^M : M \subseteq E \text{ is a perfect matching of }G}
\]
In particular, the vertices of $P_G$ correspond bijectively to perfect matchings of $G$.
It is well-known that $P_G$ has the following compact halfspace-encoding (see e.g.\ \cite{korte2011combinatorial}). Here the term $\delta(v)$ denotes the set of edges incident to vertex $v$.
\begin{align*}
    \sum_{e \in \delta(v)} x_e &= 1 &\forall v \in V \\
    x_e &\geq 0 &\forall e \in E
\end{align*}

\subsection{Our contribution}
We present two main results. 
As our first main, result, we decidedly answer Frieze and Tengs' 30 year old question. 
Let $\textsc{BPM-Diameter-Decision}$ denote the restriction of \textsc{Diameter-Decision} to only the bipartite perfect matching polytope, then
\begin{theorem}
\label{thm:main-thm-pi-2}
    $\textsc{BPM-Diameter-Decision}$ is $\Pi^p_2$-complete.
\end{theorem}

In particular, computing the diameter of a polytope, i.e.\ problem \textsc{Diameter}, is $\Pi^p_2$-hard, even when restricted to totally unimodular constraint matrices.
Since $\Pi^p_2 \supseteq \text{NP}$ is a larger complexity class than the class NP (and this inclusion is strict under the complexity-theoretic assumption that the polynomial hierarchy does not collapse to the first level), this result is stronger than the hardness results of \cite{DBLP:journals/cc/FriezeT94,DBLP:conf/focs/Sanita18,nobel2025complexity}.

Is it also true that  $\textsc{Diameter-Decision} \in \Pi^p_2$ for all families of polytopes beside the bipartite perfect matching polytope? 
A standard argument shows that for all families $\mathcal{F}$ of polytpes where it is known that the polynomial Hirsch conjecture holds, we indeed have $\textsc{Diameter-Decision} \in \Pi^p_2$ restricted to that family $\mathcal{F}$.
Hence if the polynomial Hirsch conjecture is true, we have that \textsc{Diameter-Decision} is $\Pi^p_2$-complete and our hardness result cannot be improved further (than the class $\Pi^p_2$).
Even without the assumption that the polynomial Hirsch conjecture is true, our hardness result cannot be improved further for the bipartite perfect matching polytope, and the family of all 0/1-polytopes, since here the  Hirsch conjecture is already known to hold.

What are the consequences of a problem being $\Pi^p_2$-complete? Such problems are in general much harder to solve than even NP-complete problems. 
(An illustrative explanation of that fact is given e.g.\ by Woeginger \cite{DBLP:journals/4or/Woeginger21}.)
Since integer programming is contained in NP, unless the polynomial hierarchy collapses, our result implies the diameter of a polytope cannot be computed by a mixed integer program with only polynomially many variables and constraints.
As a consequence, computation of the diameter can be compared to other notoriously hard-to-compute parameters, like the generalized Ramsey number (which Schaefer showed to be $\Pi^p_2$-complete \cite{DBLP:journals/jcss/Schaefer01}).

Since exact computation of the diameter turns out to be very hard, we turn to approximation of the diameter and prove as our second main result

\begin{restatable}{theorem}{thmInapprox}
\label{thm:inapproximability}
    There exists $\varepsilon >0$ such that the diameter of the bipartite perfect matching polytope cannot be approximated better than $(1 + \varepsilon)$ in polynomial time, unless $P \neq NP$. 
\end{restatable}

%Concretely, we prove \cref{thm:inapproximability} for $\varepsilon \approx 6.1 \cdot 10^{-5}$. 
Our proof of \cref{thm:inapproximability} is an adaptation of a recent proof of Nöbel and Steiner \cite{nobel2025complexity}.
Interestingly, Nöbel and Steiner suspected that their methods could probably not be adapted to show an inapproximability result. 
Hence we show that their suspicion is not the case, if only in the weak sense of a $(1 + \varepsilon)$-approximation.
Sanità showed that finding the diameter of the fractional matching polytope is an APX-hard problem \cite{DBLP:conf/focs/Sanita18}.
%(when interpreted as the minimization problem of finding two vertices of small distance).
Hence our \cref{thm:inapproximability} can be understood as an extension of Sanità's result to the bipartite perfect matching polytope (and thus the totally unimodular case).

\subsection{Applications to the circuit diameter}
In 2015, Borgwardt, Finhold and Hemmecke \cite{DBLP:journals/siamdm/BorgwardtFH15} proposed a generalization of the diameter of a polytope, called the \emph{circuit diameter}.
This concept is motivated by so-called \emph{circuit augmentation schemes}, which are generalizations of the simplex algorithm and have attracted a lot of attention over the last decade.
The basic idea behind circuit augmentation schemes is that we would like to move more freely through some polytope $P$ than only traversing its edges. 
These algorithms consider a set $\mathcal{C}(P)$ of potential directions, called the \emph{circuits} of the polytope. 
For the sake of completeness, we include here the definition of the set $\mathcal{C}(P)$ (although it is not needed to understand the rest of the paper).
\begin{definition*}[from \cite{borgwardt2025hardness}]
    Let $P = \set{x \in \R^d : Ax = a, Bx \leq b}$ be a polyhedron. The set of circuits of $P$, denoted by $\mathcal{C}(P)$ with respect to its linear description $A,B$ consists of all vectors $g \in \text{ker}(A) \setminus \set{0}$ for which $Bg$ is support-minimal in the set
    $\set{Bx : x \in \text{ker}(A) \setminus \set{0}}$ and $g$ has coprime integral components.
\end{definition*}
It can be shown that $\mathcal{C}(P)$ always includes the original edge directions of $P$. A single \emph{circuit move} starting from some $x \in P$ now selects some circuit $g \in \mathcal{C}(P)$ and moves maximally along direction $g$, 
i.e.\ to the point $x + \lambda g$ with the property that $\lambda \geq 0$ is maximal with respect to $x + \lambda g \in P$.
Hence circuit moves are allowed to cross the interior of the polytope and end up at non-vertices of $P$.
This fact makes circuit augmentation schemes more powerful than the simplex algorithm and a compelling subject of study. 
Indeed, many different pivot rules and algorithms for circuit augmentation schemes have been proposed by various authors over the last 10 years \cite{de2015augmentation,borgwardt2020implementation,borgwardt2021note,DBLP:journals/siamjo/LoeraKS22,dadush2022circuit}.
Analogous to the regular distance and diameter, we can define the \emph{circuit distance} between $x,y \in P$ as the minimum number of circuit moves necessary to reach $y$ from $x$, and the \emph{circuit diameter} $\cdiam(P)$ as the maximum circuit distance over all vertex pairs of $P$.
Analogously to the simplex algorithm, $\cdiam(P)$ gives a lower bound for the worst-case runtime of circuit augmentation schemes under \emph{any} pivot rule.
Remarkably, De Loeira, Kafer and Sanità came very close to proving the polynomial \enquote{circuit} Hirsch conjecture, in the sense that they proved that for every polytope $P$, 
its circuit diameter is bounded by a polynomial in $n, d$ and the maximum encoding-length among the coefficients in its input encoding \cite{DBLP:journals/siamjo/LoeraKS22}.
The classic Hirsch conjecture adapted to the circuit diameter, as proposed by  Borgwardt, Finhold and Hemmecke \cite{DBLP:journals/siamdm/BorgwardtFH15} still remains open.
In order to shed light on this question, multiple authors have studied the circuit diameter of different families of polytopes, e.g.\ in \cite{stephen2015circuit,borgwardt2018circuit,kafer2019circuit,dadush2022circuit,black2023circuit}.

The question about the computational complexity of computing the circuit diameter was raised by Sanità \cite{sanitaSummerSchool} and reiterated by Kafer \cite{kafer2022polyhedral} and Borgwardt et al.\ \cite{borgwardt2025hardness}.
This question was recently (partially) settled by Nöbel and Steiner \cite{nobel2025complexity}, who showed strong NP-hardness.
Analogously to the circuit diameter, they could not show containment in NP.

The main results in this paper have interesting consequences for the circuit diameter. 
This is because all our results are obtained for the perfect matching polytope, and for this polytope it is actually known that $\diam(P) = \cdiam(P)$ (see e.g.\ Corollary~4 in \cite{DBLP:journals/disopt/BorgwardtV22} or Lemma~2 in \cite{DBLP:conf/ipco/CardinalS23}).
Hence we immediately obtain the following corollaries of our main theorems. 
Let \textsc{Cdiam-Decision} denote the decision problem corresponding to computing the circuit diameter, then

\begin{cor}
\textsc{Cdiam-Decision} is $\Pi^p_2$-complete, even when restricted to the bipartite perfect matching polytope. 
\end{cor}

This strengthens the recent NP-hardness result of Nöbel and Steiner \cite{nobel2025complexity} and shows that just like the diameter, the circuit diameter is even harder to compute than NP-hard.
We remark that in contrast to \cref{thm:main-thm-pi-2}, we automatically have the containment  $\textsc{Cdiam-Decision} \in \Pi^p_2$ due to the almost polynomial \enquote{circuit} Hirsch conjecture of \cite{DBLP:journals/siamjo/LoeraKS22}. 
Hence this hardness result cannot be strengthened further than the class $\Pi^p_2$. As a corollary of our second main theorem, we have

\begin{cor}
\label{corollary:circuit-inapprox}
    There exists $\varepsilon > 0$, such that the circuit diameter of a polytope cannot be approximated by a factor better than $(1 + \varepsilon)$, unless P = NP.
\end{cor}

\cref{corollary:circuit-inapprox} provides the first known inapproximability result for the circuit diameter of a polytope.
As a consequence, unless P = NP, no circuit augmentation scheme can find the optimal solution in worst-case runtime of $(1 + \varepsilon)\cdiam(P)$ steps or less. 
This can be compared to the result of Cardinal and Steiner \cite{DBLP:conf/ipco/CardinalS23}, 
who show that no circuit augmentation scheme can have worst-case runtime $O(\log n / \log \log n)$ to find a circuit walk between two vertices of distance two (assuming ETH).
%where $\cdiam(P)$ is a trivial lower bound for the worst-case runtime of any such scheme. 
It is a very interesting question to ask whether our inapproximability result can be improved to any constant factor.

\subsection{Related work}
There are several other works that deal not with the diameter, but rather with the problem of computing the \emph{distance} of two given vertices of some polytope $P$.
In the case of the distance this problem was shown to be NP-hard for the bipartite perfect matching polytope $P_G$ independently by Aichholzer et al.\ \cite{DBLP:journals/algorithmica/AichholzerCHKMS21} and Ito et al.\ \cite{DBLP:journals/siamdm/ItoKKKO22}. Additionally, Borgward et al.\ \cite{borgwardt2025hardness} show NP-hardness of computing circuit distances in 0/1-network flow polytopes, and Cardinal and Steiner show the same for polymatroids \cite{cardinal2023shortest}.
These results were strengthened by Cardinal and Steiner \cite{DBLP:conf/ipco/CardinalS23} to a $O(\log n / \log \log n)$-hardness of approximation result for computing a short path between two 
vertices of distance two on the bipartite perfect matching polytope, assuming ETH.
As explained in \cite{nobel2025complexity}, the techniques for handling polytope diameters and distances are very different, since distances can be very short (which is often helpful for proofs), 
but diameters are usually large.

Finally, a related series of work deals with \emph{combinatorial reconfiguration} of perfect matchings \cite{DBLP:journals/tcs/ItoDHPSUU11,kaminski2012complexity,DBLP:conf/sofsem/GuptaKM19,DBLP:conf/wg/BousquetHIM19,bonamy2019perfect, binucci2025flipping}. 
In these papers, the notion of adjacency between two perfect matchings differs from the adjacency on the bipartite perfect matching polytope considered in the present paper.

\section{Technical Overview}
\label{sec:technical-overview}

We give a quick overview of the arguments used to obtain the main result of the paper. Recall the following.
A perfect matching (abbreviated PM) of a graph $G = (V, E)$ is a subset $M \subseteq E$ such that every vertex is touched exactly once.
A path or cycle is \emph{alternating} with respect to some PM $M$ if it alternates between edges in $M$ and not in $M$.
An alternating path is explicitly allowed to start and/or end with an edge not in $M$.
It is well-known that given two PMs $M_1, M_2$, their \emph{symmetric difference} $M_1 \symdiff M_2 = (M_1 \setminus M_2) \cup (M_2 \setminus M_1)$ is a vertex-disjoint union of alternating cycles.
Since we are concerned with shortest paths in the bipartite perfect matching polytope, we first need to understand, when two of its vertices are adjacent. This is a well-known fact first shown by Chvátal.
\begin{lemma*}[\cite{chvatal1975certain}, see also \cite{nobel2025complexity}]
Let $G = (V,E)$ be a bipartite graph, and let $x = \chi^M$ and $y = \chi^N$ be two vertices of the polytope $P_G$, corresponding to two perfect matchings $M,N \subseteq E$.
The vertices $x$ and $y$ are adjacent in the 1-skeleton of $P$ if and only if the edges in the symmetric difference $M \symdiff N$ form a single alternating cycle.
\end{lemma*}

Based on this insight, we define the following:
\emph{Flipping} an alternating cycle $C$ means transforming the PM $M$ into the PM $M \symdiff C$.
A \emph{flip sequence} of length $d$ from a PM $M_1$ to a PM $M_2$ is a sequence $(C_1, \dots, C_d)$ such that
for all $i \in [d]$, the cycle $C_i$ is alternating with respect to the perfect matching $M \symdiff C_1 \symdiff \dots \symdiff C_{i-1}$, and we have $M_2 = M_1 \symdiff C_1 \symdiff \dots \symdiff C_d$.
For two PMs $M_1, M_2$, we define $\dist(M_1, M_2)$ as the length of the shortest flip sequence from $M_1$ to $M_2$.
Finally, we define problem \textsc{BPM-Diameter-Decision} as the restriction of problem \textsc{Diameter-Decision} to the special case of the bipartite perfect matching polytope. 
This means that the input is a bipartite graph $G$, and some threshold $t \in \N$. The problem is to decide whether $\diam(P_G) \leq t$, or alternatively whether
\[
\max\set{\dist(M_1, M_2) : M_1, M_2 \text{ PM of }G} \leq t.
\]

\textbf{Overview of the first main result.}
We wish to prove that \textsc{BPM-Diameter-Decision} is $\Pi^p_2$-complete. 
As a starting point, one can look at the classic $\Pi^p_2$-complete problem $\forall \exists$-SAT \cite{DBLP:journals/tcs/Stockmeyer76}. 
Here, we are given a SAT-formula $\varphi$ in conjunctive normal form, such that its variables are partitioned into two disjoint parts $X, Y$ (and such that w.l.o.g.\ every clause contains exactly three literals). The question is whether for all assignments of $X$ there exists an assignment of $Y$ that makes $\varphi$ true, i.e.\
\[
\forall \alpha \in \set{0,1}^X : \exists \beta \in \set{0,1}^Y: \varphi(\alpha, \beta) = 1.
\]
One can compare this to the definition of the diameter. The diameter of a polytope $P$ is at most $d$ if and only if for all vertex pairs of $P$ there is a path of length at most $d$ in the 1-skeleton connecting those two vertices. 
In the special case of the bipartite perfect matching polytope $P_G$ of some graph $G$ this means that
\begin{equation}
    \forall M_1, M_2 \text{ a PM of $G$ } : \exists \text{ a flip sequence of length at most $d$ from $M_1$ to $M_2$}.
    \label{eq:diameter}
\end{equation}

Our main idea is to combine recent advances in understanding of both the diameter of the bipartite perfect matching polytope, as well as recent understandings of the typical structure of $\Pi^p_2$-complete problems.
In particular, the clever construction of Nöbel and Steiner \cite{nobel2025complexity} manages to transform a given graph $H$ into another graph $G$ such that short flip sequences in $G$ relate closely to Hamiltonian cycles in $H$.
They use this to obtain an elegant proof for the NP-hardness of \textsc{Diameter}.
On the other hand, Grüne and Wulf \cite{grune2024completeness} recently introduced a wide-ranging meta-theorem that characterizes a \enquote{typical} problem complete for the class $\Sigma^p_2$ = co-$\Pi^p_2$. (A similar result is also presented in \cite{grune2025complexity}.) 
The study of \cite{grune2024completeness} motivates the definition of the following $\Pi^p_2$-complete variant of the Hamiltonian cycle problem.
Let $k \in \N$. Given a directed graph $H$ together with a set $E' = \set{e_1, \overline e_1, \dots e_k, \overline{e_k}}$ of $2k$ distinct arcs, consider a subset $P \subseteq E'$. We say that $P$ is a \emph{pattern}, if $|P \cap \set{e_i, \overline{e_i}}| = 1$ for all $i \in [k]$. 
Some Hamiltonian cycle $C$ in $H$ \emph{respects} the pattern $P$, if $C \cap E' = P$. In other words, the pattern determines for each pair $\set{e_i, \overline{e_i}}$, which one of these two arcs the Hamiltonian cycle is allowed to use and which one it is not allowed to use.
We prove in \cref{sec:base-problem} using techniques from \cite{grune2024completeness} that the following problem is $\Pi^p_2$-complete:
\begin{quote}
    \textbf{Problem}: $\forall \exists$-$\textsc{HamCycle}$

    \textbf{Input}: Directed graph $H$, distinct vertices $v^{(1)}, \dots, v^{(k)}$ for some $k \in \N$, such that for all $i \in [k]$, the vertex $v^{(i)}$ has exactly two outgoing arcs, denoted by $e_i, \overline e_i$. The set $E' := \set{e_1, \overline e_1, \dots, e_k, \overline e_k}$. 

    \textbf{Question}: Is the following true:\\
    $\forall \text{ patterns } P \subseteq E' : \exists \text{ a Hamiltonian cycle $C$ respecting $P$}$?
\end{quote}

We reduce $\forall \exists$-$\textsc{HamCycle}$ to \textsc{BPM-Diameter-Decision}. 
In \cref{sec:combining-the-pieces}, we show:
Given an instance of $\forall \exists$-$\textsc{HamCycle}$, consisting out of a directed graph $H$ on $n = |V(H)|$ vertices together with arcs $E' = \set{e_1, \overline e_1, \dots e_k, \overline{e_k}}$, one can in polynomial time construct an undirected bipartite graph $G_H$ on $\Theta(n^9)$ vertices such that the following two lemmas hold:

\begin{restatable}{lemma}{mainTheoremIf}
\label{lem:main-thm-if}
If for all patterns $P \subseteq E'$ the directed graph $H$ contains a Hamiltonian cycle respecting $P$, then $\diam(P_{G_H}) \leq 4n^4 + 46n$.
\end{restatable}

\begin{restatable}{lemma}{mainTheoremOnlyIf}
\label{lem:main-thm-only-if}
If $\diam(P_{G_H}) \leq 4n^4 + 46n$, then for all patterns $P \subseteq E'$ the directed graph $H$ contains a Hamiltonian cycle respecting $P$.
\end{restatable}

Together, these two lemmas prove our main theorem. In order to understand how our construction works, it is helpful to first understand the elegant idea of \cite{nobel2025complexity}.
Nöbel and Steiner introduce so-called \emph{tower-gadets} (see \cref{fig:tower-gadget}), and furthermore combine multiple towers into 
a gadget which for our purpose shall be called \emph{city gadget} (see \cref{fig:city-gadget}). 
They show that there exist perfect matchings $M_1$ and $M_2$ such that any short flip sequence from $M_1$ to $M_2$ must contain at least one cycle that visits every city gadget. This is sufficient for their argument, but we require a stronger statement. Let us call a cycle $C$ in the graph $G_H$ \emph{regular}, if $C$ visits every city gadget. 

We first show, that in a short enough flip sequence $(C_1, \dots, C_d)$, if $n$ tends to infinity, then an arbitrarily large percentage of the cycles $(C_1, \dots, C_d)$ must be regular.
We remark that it is inherent to the arguments used in \cite{nobel2025complexity} that in general there can always be several cycles that are irregular, i.e.\ they do not visit every city gagdet.
Such irregular cycles are technically challenging for us to deal with, but their existence can not be excluded.

We significantly expand upon the ideas of \cite{nobel2025complexity} by introducing so-called \emph{XOR-gadgets} (\cref{fig:xor-gadget}), 
\emph{ladder-gadgets} (\cref{fig:ladder-gadget}) and sophisticated \emph{$\forall$-gadgets} (\cref{fig:forall-gadget-definition}). 
These gadgets have the property that if they interact with some regular cycle $C$, then their behavior is well understood. 
On the other hand, if they interact with some irregular cycle, then the \enquote{damage} done by this irregular cycle is in some sense restricted.

It is described in \cref{sec:combining-the-pieces} how to define the desired bipartite undirected graph $G_H$ given the directed graph $H$. The graph $G_H$ contains many copies of the gadgets described above.
Given two perfect matchings $M_1, M_2$ of $G_H$, one can ask what is their distance $\dist(M_1, M_2)$ in the polytope $P_{G_H}$. 
We are able to show, that in a certain sense, this distance crucially depends on the structure of the matchings $M_1, M_2$ restricted to the $\forall$-gadgets.
In particular, consider a fixed pattern $P \subseteq E'$ in $H$. 
We show that one can find PMs $M_1, M_2$ in $G_H$ by carefully choosing their restriction to the $\forall$-gadgets, 
and completing this choice to the whole graph $G_H$ in the right manner, such that $\dist(M_1, M_2)$ is small if and only if there is a Hamiltonian cycle respecting the pattern $P$ in the directed graph $H$.
(This argument depends crucially on the asymmetric behavior of the set of perfect matchings of the ladder gadget, explained more closely in \cref{sec:ladders}.)
Note that this approach highlights the similarity between Problem~(\ref{eq:diameter}) and $\forall \exists$-\textsc{HamCycle}, i.e.\ perfect matchings of $G_H$ correspond to patterns $P \subseteq E'$ and vice versa.
However, note that there is a major difference: 
Patterns are binary, that is for all $i \in [k]$ the set $P$ either contains $e_i$ or $\overline e_i$. 
In contrast, the set of perfect matchings of a $\forall$-gadget is not binary at all, rather it is exponentially large.
Nonetheless, we are able to overcome this final technical difficulty. On a high level, 
we show that we can interpret a PM of a $\forall$-gadget as a \enquote{mixture} of so-called top-states and bottom-states, 
and make use of the property that \emph{for all} patterns $P \subseteq E'$ the directed graph $H$ has a Hamiltonian cycle respecting $P$.

We explain one more idea, which is also used in \cite{nobel2025complexity}.
We will define a subset $V_s \subseteq V(G_H)$ of size $|V_s| = \Theta(n)$ of the $\Theta(n^9)$ vertices of $G_H$, called the \emph{semi-default vertices}.
We call a PM $M$ of the graph $G_H$ a \emph{semi-default matching}, if  $M$ matches the semi-default vertices in a pre-described way.
We show that the semi-default matchings lie \enquote{dense} in all perfect matchings. In particular, given an arbitrary PM $M$, one can find a semi-default PM $M'$ such that $\dist(M, M') \leq 23n \ll n^4$.
The intuition we will use is that in order to transform any PM $M_1$ into any PM $M_2$, we first look for semi-default representatives $M_1', M_2'$ of $M_1, M_2$. 
Then, we use at most $23n$ irregular cycles to transform $M_1$ to $M_1'$, at most $4n^4$ regular cycles to transform $M'_1$ into $M'_2$, and at most $23n$ irregular cycles to transform $M'_2$ into $M_2$.
This completes the proof sketch of the first main result. More details are given in \cref{sec:reduction}.

\textbf{Overview of the second main result.} We wish to prove that the diameter is hard to approximate to a factor $(1 + \varepsilon)$.
For this, we consider a similar reduction as for the first main result.
We again make central use of the city gadgets introduced by \cite{nobel2025complexity} to \enquote{enforce} that most cycles in a short flip sequence $(C_1,\dots,C_d)$ visit all cities.
This is reminiscent of the Hamiltonian cycle problem.
It is therefore natural to show a certain approximation version of the Hamiltonian cycle problem. 
It roughly states that we can w.l.o.g.\ consider problem instances such that in yes-instances it is possible to visit all cities, but in no-instances it is only possible to visit a $(1 - \varepsilon)$ fraction of all cities.
We prove a result of this kind suited to our needs in \cref{thm:eps-good-reduction}. We achieve this as a consequence of the PCP theorem using established techniques.
If we have a yes-instance, then the diameter of $P_G$ is small, since we can efficiently visit all city gagdets.
However, if it is only possible to visit a $(1 - \varepsilon)$ fraction of all cities, then the diameter of $P_G$ must be large. 
This is because city gadgets (in the right configuration) have the property that each city must be visited many times by the cycles in $(C_1, \dots, C_d)$. 
This completes the proof sketch of the second main result. More details are given in \cref{sec:inapprox}.

\section{Preliminaries}
\label{sec:prelim}

In this paper, we use the notation $\N = \set{1,2,\dots}$ and $[n] := \fromto{1}{n}$. We use $uv$ as a shorthand notation for the undirected edge $\set{u,v}$. All paths and cycles are simple. All graphs  are assumed to be undirected, unless stated otherwise. Paths and cycles in a directed graph always adhere to the orientation of the edges. For $G = (V, E)$, we denote $V(G) := V$ and $E(G) := E$.

A \emph{language} is a set $L\subseteq \set{0,1}^*$.
A language $L \subseteq \set{0,1}^*$ is contained in $\Pi^p_2$ if there exists some polynomial-time computable function $V$ (verifier), such that for all $w \in \set{0,1}^*$ for suitable $m_1,m_2 \leq \text{poly}(|w|)$
$$
    w \in L \ \Leftrightarrow \ \forall y_1 \in \set{0,1}^{m_1} \ \exists y_2 \in \set{0,1}^{m_2}: V(w,y_1,y_2) = 1.
$$
A language $L$ is $\Pi^p_2$-hard, if every $L' \in \Pi^p_2$ can be reduced to $L$ with a polynomial-time many-one reduction. If $L$ is both $\Pi^p_2$-hard and contained in $\Pi^p_2$, it is $\Pi^p_2$-complete.

An introduction to the polynomial hierarchy and the classes $\Pi^p_2$ and $\Sigma^p_2 = \text{co-}\Pi^p_2$ can be found in the book by Papadimitriou \cite{DBLP:books/daglib/0072413}, or the article by Woeginger \cite{DBLP:journals/4or/Woeginger21}.

Let $X = \set{x_1,\dots,x_n}$ be a set of boolean variables. The corresponding set of \emph{literals} is $\fromto{x_1}{x_n} \cup \fromto{\overline x_1}{\overline x_n}$.
A \emph{clause} is a disjunction of literals. A boolean formula $\varphi$ is in conjunctive normal form (CNF) if it is a conjunction of clauses.
An \emph{assignment} is a map $\alpha : X \to \set{0,1}$. 
We write $\varphi(\alpha) \in \set{0,1}$ to denote the evaluation of  $\varphi$ under assignment $\alpha$.

\section{The Reduction}
\label{sec:reduction}
In this section, we present our proof of $\Pi^p_2$-completeness. We first show containment in the class $\Pi^p_2$ in \cref{subsec:containment}. Then we introduce the $\Pi^p_2$-complete problem $\forall \exists$-\textsc{HamCycle} in \cref{sec:base-problem}. 
We introduce a number of different gadgets in \cref{sec:towers-and-cities,sec:ladders}, 
with the $\forall$-gagdet being the most complicated one. In \cref{sec:combining-the-pieces} we define the graph $G_H$, and in \cref{sec:proof-of-main-thm}, 
we show that the diameter of the polytope $P_{G_H}$ has the desired properties.

\subsection{Containment}
\label{subsec:containment}
It is relatively easy to show the containment of \textsc{Diameter-Decision} in the class $\Pi^p_2$, 
as long as the following assumption holds.
\begin{lemma}
    \label{lem:containment}
    Let $\mathcal{F}$ be a family of polytopes. 
    If it holds that for each $P \in \mathcal{F}$ we have $\diam(P) \leq \text{poly}(n,d,|P|)$, where $n$ is its number of facets, 
    $d$ its dimension, and $|P|$ its encoding length, then \textsc{Diameter-Decision} restricted to the family $\mathcal{F}$ is contained in $\Pi^p_2$.
    An analogous statement holds for the circuit diameter $\cdiam(P)$ and \textsc{Cdiam-Decision}.
\end{lemma}
\begin{proof}
    Let $P = \set{x \in \R^d : Ax \leq b}$, and consider some $t \in \N$.
    We have $\diam(P) \leq t$ if and only if
    \[
    \forall x,y \text{ vertices of $P$ } : \exists (v_1,\dots,v_k) : x=v_1, y=v_k, k \leq t, \text{all $v_i$  and $v_{i+1}$ are neighboring vertices.} 
    \]
    Furthermore, given a sequence $(v_1,\dots,v_k)$ of vectors in $\R^d$, 
    it is possible to test in polynomial time, whether $v_1,\dots,v_k$ 
    are actually vertices of the polytope and whether $v_i$ is adjacent to $v_{i+1}$ for all $i \in [k-1]$ (see e.g.\ \cite{DBLP:journals/cc/FriezeT94}).
    In particular, we can w.l.o.g.\ restrict such sequences to length $k \leq \diam(P) \leq \text{poly}(n,d,|P|)$.
    Therefore, the encoding space of the sequence $(v_1,\dots,v_k)$ is at most polynomial in the input size.
    Thus by the definition of $\Pi^p_2$, we have $\textsc{Diameter-Decision} \in \Pi^p_2$.
    For the circuit diameter, the same argument applies. We use the fact that given two points $v_i,v_{i+1} \in P$, one can test in polynomial time whether $v_i,v_{i+1}$ are one circuit move apart \cite{borgwardt2025hardness}.
\end{proof}

In particular, we have $\textsc{BPM-Diameter-Decision} \in \Pi^p_2$, and $\textsc{CDiam-Decision} \in \Pi^p_2$.

\subsection{The Base Problem}
\label{sec:base-problem}

Consider problem $\forall \exists$-\textsc{HamCycle}, as defined in \cref{sec:technical-overview}.

\begin{lemma}
$\forall \exists$-\textsc{HamCycle} is $\Pi^p_2$-complete (even when restricted to planar directed graphs of maximum in- and outdegree 2).
\end{lemma}
\begin{proof}
    The problem is contained in $\Pi^p_2$, since it can be stated as: 
    \[
    \forall \text{ patterns } P \subseteq E' : \exists \text{ a Hamiltonian cycle $C$ respecting $P$},
    \]
    and given a directed graph $H$ together with sets $P \subseteq E'$, $C \subseteq E(H)$, it can be checked in polynomial time whether $C$ is a Hamiltonian cycle and respects the pattern $P$. So it remains to show $\Pi^p_2$-hardness.
    Our proof strategy is to first consider a classic reduction by Plesn{\'i}k \cite{plesn1979np}, which reduces 3SAT to the directed Hamiltonian cycle problem in planar digraphs of maximum in- and outdegree 2. 
    We show that we can modify Plesn{\'i}k's reduction to be instead a reduction from $\forall \exists$-SAT to $\forall \exists$-\textsc{HamCycle}.
    It is not necessary to re-iterate all the details of Plesn{\'i}k's reduction. Instead, it suffices to notice the following key properties: 
    Given a 3SAT formula $\varphi$, the reduction constructs in polynomial time a (planar and 2-in/outdegree-bounded) directed graph $H_\varphi$.  
    If $\varphi$ has $N$ variables and $M$ clauses, the graph $H_\varphi$ contains $N$ so-called variable gadgets and $M$ so-called clause gadgets. 
    Plesn{\'i}k proceeds to prove that every Hamiltonian cycle must traverse the variable gadgets in order, and must make a \enquote{decision} for each variable gadget.
    In particular, the $i$-th variable gadget starts with some vertex $v^{(i)}$ of outdegree 2. 
    Immediately after entering $v^{(i)}$, every Hamiltonian cycle needs to decide whether it leaves $v^{(i)}$ via the first outgoing arc and traverses the variable gadget in the \enquote{$x_i = \text{true}$} configuration,
    or if it leaves $v^{(i)}$ via the second outgoing arc and traverses the variable gadget in the \enquote{$x_i = \text{false}$} configuration.
    Let us define $e_i$ as the first, and $\overline e_i$ as the second outgoing arc of vertex $v^{(i)}$, and let $E'' = \set{e_1, \overline e_1, \dots, e_N, \overline e_N}$. 
    By studying \cite{plesn1979np} it can be checked that the following holds:

    \begin{itemize}
    \item For every $i \in \fromto{1}{N}$, every Hamiltonian cycle $C$ in $H_\varphi$ uses exactly one of the two arcs $e_i$ and $\overline e_i$.
    \item There is a direct correspondence between subsets of the arc set $E''$ that can be part of a Hamiltonian cycle and satisfying assignments for the formula $\varphi$.
    Formally: Let $Z = \fromto{z_1}{z_N}$ be the variable set of $\varphi$. For an assignment $\alpha : Z \to \set{0,1}$, let $P_\alpha$ denote the corresponding set of arcs defined by $P_\alpha = \set{e_i : \alpha(z_i) = 1} \cup \set{\overline e_i : \alpha(z_i) = 0}$. 
    Given an assignment $\alpha$, there exists a Hamiltonian cyle $C$ in $H_\varphi$ with $C \cap E'' = P_\alpha$ if and only if $\varphi(\alpha) = 1$.
    \end{itemize}
    Finally, we describe how to modify this reduction to a reduction from $\forall \exists$-SAT to $\forall \exists$-\textsc{HamCycle}.
    Given a 3SAT formula $\varphi(X,Y)$, we assume w.l.o.g.\ that $\varphi$ is a formula over $N = 2n$ variables partitioned into two parts $X,Y$ of equal size $|X| = |Y| = n$. We consider the graph $H_\varphi$ as above.
    Then, we consider the arcs $E'' = \set{e_1, \overline e_1, \dots, e_{2n}, \overline e_{2n}}$, 
    and let $E' = \set{e_1, \overline e_1, \dots, e_n, \overline e_n}$ be those arcs of $E''$ that appear in the variable gadgets corresponding to $X$ (but not to $Y)$.
    By the above observation, for a fixed pattern $P \subseteq E'$, there is a Hamiltonian cycle respecting $P$, 
    if and only if it is possible to complete $P$ with another pattern $P_2 \subseteq E'' \setminus E'$ and complete $P \cup P_2$ to a Hamiltonian cycle.
    Again, by the above correspondence, this implies that $\varphi$ is a yes-instance of $\forall \exists$-SAT if and only if $(H_\varphi, E')$ is a yes-instance of $\forall \exists$-\textsc{HamCycle}.
\end{proof}

We remark that the key properties described above are not specific to Plesn{\'i}k's proof. Similar properties appear in other folklore proofs for the NP-hardness of Hamiltonian cycle, and in plenty of NP-hardness proofs of other combinatorial problems. 
In \cite{grune2024completeness} a systematic study of this property is performed, and its implications on completeness in the second level of the polynomial hierarchy are discussed. 

\subsection{Towers and Cities}
\label{sec:towers-and-cities}

\begin{figure}
    \centering
    \begin{subfigure}{0.2\textwidth}
        \includegraphics[scale=1,page=1]{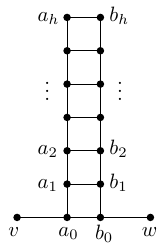}
        \caption{Definition}
    \label{fig:tower-gadget:definition}
    \end{subfigure}
    \hfill
    \begin{subfigure}{0.2\textwidth}
        \includegraphics[scale=1,page=2]{img/tower-gadget.pdf}
        \caption{default}
    \label{fig:tower-gadget:default}
    \end{subfigure}
    \hfill
    \begin{subfigure}{0.2\textwidth}
        \includegraphics[scale=1,page=3]{img/tower-gadget.pdf}
        \caption{locked}
    \label{fig:tower-gadget:locked}
    \end{subfigure}
    \hfill
    \begin{subfigure}{0.2\textwidth}
        \includegraphics[scale=1,page=4]{img/tower-gadget.pdf}
        \caption{semi-default}
    \label{fig:tower-gadget:semi-default}
    \end{subfigure}
    \hfill
    \caption{A tower gadget of height $h$, and different states a perfect matching of the tower gadget may have.}
    \label{fig:tower-gadget}
\end{figure}
%-----------------------------------
Nöbel and Steiner \cite{nobel2025complexity} introduce so-called \emph{tower gadgets}. For some $h \in \N$, a tower gadget of height $h$ is an induced subgraph on the vertex set
\[
\set{v,w} \cup \bigcup_{i=0}^h \set{a_i, b_i}
\]
and edge set 
\[
\set{va_0, b_0w} \cup \bigcup_{i=0}^h \set{a_ib_i} \cup \bigcup_{i=1}^h \set{a_ia_{i-1}, b_ib_{i-1}}.
\]
An example for $h = 6$ is depicted in \cref{fig:tower-gadget:definition}. 
Recall that our plan is to use these tower gadgets as part of some larger graph $G_H$. 
We will maintain the invariant that every tower gadget is connected to the rest of $G_H$ only via the two vertices $v$ and $w$.
If some tower gadget $T$ is contained in $G_H$, and $M$ is a PM of $G_H$, the PM $M$ restricted to $T$ can behave in many different ways. We give names to some of these behaviors.
Consider \cref{fig:tower-gadget:default,fig:tower-gadget:locked,fig:tower-gadget:semi-default}.
With respect to a fixed PM $M$, the tower gadget $T$ is called in \emph{default} state, if $\set{va_0, b_0w} \cup \bigcup_{i=1}^h \set{a_ib_i} \subseteq M$.
The tower gadget $T$ is called \emph{locked}, if it is in default state, except for the 4 top vertices, where instead $a_ha_{h-1}, b_hb_{h-1} \in M$.
The tower gadget $T$ is called in \emph{semi-default state}, if $va_0, b_0w \in M$ (and the rest of the tower can be matched in any arbitrary way). 

Consider a tower gadget $T$, subgraph of $G_H$, which starts in the locked state, 
and which we want to transform into the default state using a flip sequence of alternating cycles $(C_1,\dots,C_d)$, where $C_i \subseteq G_H$ for $i \in [d]$.
Of course, we can achieve this task by using a single alternating cycle $C$ on the vertices $\set{a_h, a_{h-1}, b_{h-1}, b_h}$. 
However, if we choose to do this, it comes with a disadvantage: Even though we have achieved our goal using a single cycle $C$, the influence of the alternating cycle $C \subseteq G_H$ is limited to only $T$ (but $C$ could be useful also somewhere else).
This motivates the following definition: Let $T$ be a tower gadget in the graph $G_H$ and $M$ be a fixed PM of $G_H$.
A cycle $C \subseteq G_H$ is called \emph{well-behaved}\footnote{In \cite{nobel2025complexity} this behavior is called \enquote{touching the tower}. However, we would like to use the same name also for our more advanced gadgets later on, where the name \enquote{touching} could be misinterpreted.} for $T$, if restricted to $T$ it is a $v$-$w$-path.
A sequence of paths $(P_1, \dots, P_d)$ is a well-behaved flip sequence for $T$, if for all $i = 1, \dots, d$ the path $P_i$ is a path from $v$ to $w$ inside $T$ such that $P_i$ is alternating with respect to the matching $M \symdiff P_1 \symdiff \dots \symdiff P_{i-1}$.

\begin{lemma}[adapted from \cite{nobel2025complexity}]
\label{lem:tower-lower-bound}
    Let $T$ be a tower gadget of height $h$ for some $h \in \N$.
    If $(P_1, \dots, P_d)$ is a well-behaved flip sequence for $T$ which transforms $T$ from the locked state into the default state, then $d \geq 2h -2$.
\end{lemma}
For the sake of completeness, we provide the proof sketch of \cref{lem:tower-lower-bound}. 
Let $M_i$ be the PM after applying the first $i$ paths $P_1, \dots, P_i$ for $i=0,\dots,d$. 
Let $\mathcal{H}(M_i) := \set{j \in [h] : a_jb_j \in M_i}$ be the set of indices corresponding to the horizontal edges of $M_i$. 
Because the flip sequence is well-behaved, every path can switch between the left and the right side of the tower only a single time.
This means that  $\mathcal{H}(M_{i+1})$ results from $\mathcal{H}(M_i)$ by adding or deleting only a single item.
Furthermore, because every path $P_i$ is also alternating, the first time $P_i$ encounters some horizontal edge $a_jb_j$ of $M_{i-1}$, the path $P_i$ must include it.
Hence, in order to go from the locked to the default position, one must first remove the $h-2$ horizontal edges of $\mathcal{H}(M_0)$, and then add all $h$ horizontal edges. 
Since each requires a single path, we have $d \geq 2h-2$.

\begin{lemma}[adapted from \cite{nobel2025complexity}]
\label{lem:tower:upper-bound}
    Let $T$ be a tower gadget of height $h$, and let $M_1, M_2$ be two PMs of $T$ in semi-default position.
    Then there exists a well-behaved flip sequence $(P_1, \dots, P_d)$ which transforms $M_1$ into $M_2$ such that $d = 2h$.
\end{lemma}

Here the intuition is that we can first remove all the at most $h$ horizontal edges of $M_1$, and then insert all the at most $h$ horizontal edges of $M_2$.
It can be checked that this is possible with a well-behaved flip sequence $(P_1, \dots, P_{d'})$ for some $d' \leq 2h$, where in particular all paths are alternating with respect to the current matching.
Since both $M_1$ and $M_2$ are in semi-default state, $d'$ is an even number.
Finally, we can assume w.l.o.g.\ that $d' = 2h$. Indeed, if $d' < 2h$, we can just choose any alternating path $P$ from $v$ to $w$ (which always exists) and consider the flip sequence $(P_1, \dots, P_{d'}, P, P)$.
Since flipping the same path twice does have no effect on the matching, this new flip sequence has the same end result.

\begin{figure}
    \centering
    \includegraphics[scale=1,page=1]{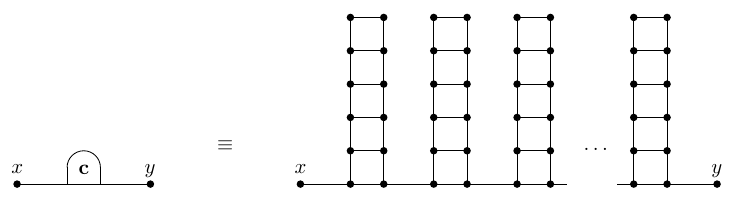}
    \caption{A city gadget consists out of many tower gadgets.}
    \label{fig:city-gadget}
\end{figure}
%-----------------------------------
\begin{figure}
    \centering
    \includegraphics[scale=1,page=2]{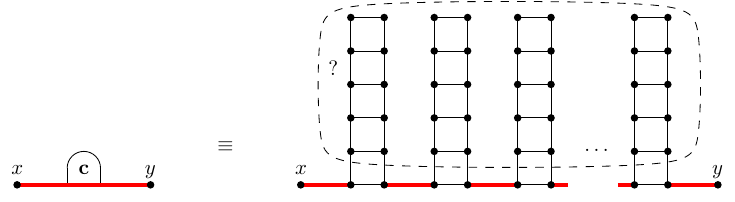}
    \caption{A city gadget is called \emph{matched}, or in \emph{semi-default state}, if all of its towers are in semi-default state.}
    \label{fig:city-gadget-matched}
\end{figure}
%-----------------------------------
We now combine multiple tower gadgets into so-called \emph{city gadgets}. These were also used in \cite{nobel2025complexity}, but not given an explicit name.
A city gadget of width $t$ and height $h$ between vertices $x$ and $y$ is defined as follows: 
It consists out of $t$ towers $T_1, \dots, T_t$, such that vertex $x$ is identified with vertex $v$ of the first tower $T_1$, vertex $y$ is identified with vertex $w$ of the last tower $T_t$, 
and for $i = 2, \dots, t$, the edge $b_0w$ of tower $T_{i-1}$ is identified with edge $va_0$ of the next tower $T_i$. An example is depicted in \cref{fig:city-gadget}.
Note that $x$ and $y$ are connected by a path of odd length.
A city gadget is called in \emph{semi-default state} or \emph{matched}, if all of its towers are in semi-default state.
Since every tower has an even number of vertices, a city gadget is matched if and only if the edge $xa_0^{(T_1)}$ of the first tower is in the matching.
Similarly to before, we will maintain the invariant, that a city gagdet is connected to the rest of $G_H$ only via the vertices $x$ and $y$.
City gagdets are useful since they \enquote{enforce} that certain locations must be visited. Let us say that some cycle $C$ in the graph $G_H$ \emph{visits a city gagdet $g$}, if the cycle $C$ contains the first edge $xa_0^{(T_1)}$ of the city gadget $g$.
Note that a cycle $C$ visits a city if and only if $C$ is well-behaved for every tower of the city.

\begin{definition}
Let $G_H$ be a graph containing several city gadgets.
A (not necessarily alternating) cycle $C$ in $G_H$ is called regular, if it visits every single city gadget of $G_H$.
\end{definition}

\begin{lemma}
\label{lem:regular-cycles}
    Let $n \in \N$. Let $G_H$ be a graph which contains at most $17n$ city gadgets, with every city gadget having the same height $h = 2n^4$ and width $t = 4n^4 + 100n$.
    Let $M_1$ be a PM of $G_H$ such that w.r.t.\ $M_1$ every tower of every city is in locked state, and $M_2$ be a PM of $G_H$ such that every tower of every city is in default state.
    If $(C_1, \dots, C_d)$ is a flip sequence from $M_1$ to $M_2$ with $d \leq 4n^4 + 46n$, then we also have $d \geq 4n^4 - 2$ and furthermore at most $1000n^2$ of the $d$ cycles are not regular.
\end{lemma}
\begin{proof}
    Consider a fixed city gadget $g$. We claim that the gagdet $g$ contains a tower $T^\star$ such that each of $C_1, \dots, C_d$ either does not visit $g$ at all or is well-behaved for $T^\star$.
    Indeed, there can be at most $d$ towers in $g$ with the property that one of the cycles $(C_1\dots,C_d)$ is entirely contained in them. Since $t > d$, at least one tower $T_j$ has the property that none of the cycles $(C_1, \dots, C_d)$ is contained entirely in $T_j$. We let $T^\star := T_j$
    Now, observe that $T^\star$ starts in locked state, ends up in default state, and only interacts with well-behaved cycles.
    Due to \cref{lem:tower-lower-bound}, at least $2h-2$ well-behaved cycles are required for this task. 
    Hence at least $2h-2$ cycles visit the city $g$. In particular $d \geq 2h - 2 = 4n^4 - 2$. If we define the set
    \[
    \mathcal{C}_g := \set{i \in [d] : C_i \text{ does not visit }g},
    \]
    then $|\mathcal{C}_g| \leq d - (2h -2) \leq 46n + 2$. Furthermore if a cycle is not regular, then it is contained in $\mathcal{C}_g$ for at least one city $g$, and so
    \[
    |\set{i \in [d] : C_i \text{ not regular}}| \leq \sum_g |\mathcal{C}_g| \leq 17n(48n) \leq 1000n^2.
    \]
    Here the index $g$ runs over all cities of the graph.
\end{proof}

\textbf{The XOR-gadget.} In the final construction of the graph $G_H$, we will choose the parameters such that \cref{lem:regular-cycles} is applicable. 
The lemma then tells us that of the $\Theta(n^4)$ cycles in any short flip sequence, at most $O(n^2)$ are not regular, i.e.\ the majority of cycles visit all cities.
Based on this insight, we introduce our next gadgdet, the so-called \emph{XOR-gadget}. Given a graph $G$, and two edges $ab, uv \in E(G)$, 
a XOR-gadget between $ab$ and $uv$ is created by subdividing the edge $ab$ four times, creating four new vertices $x_1, \dots, x_4$, and subdividing the edge $uv$ four times, creating four new vertices $y_1, \dots, y_4$. 
After that, we connect $x_i$ and $y_i$ with a city gagdet for $i=1,\dots,4$. An example is shown in \cref{fig:xor-gadget}.
\begin{figure}
    \centering
    \includegraphics[scale=1]{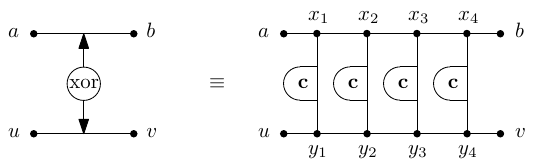}\\
    \vspace{0.4cm}
    \includegraphics[scale=1,page=2]{img/xor-gadget-states.pdf}
    \hspace{2cm}
    \includegraphics[scale=1,page=3]{img/xor-gadget-states.pdf}
    \caption{Definition of a XOR-gadget. Note that if a cycle visits every city, then in every XOR-gadget it \enquote{traverses} exactly one of the two edges $ab$ and $uv$.}
    \label{fig:xor-gadget}
\end{figure}
Note that if $C$ is a regular cycle, then $C$ traverses the XOR-gadget either from $u$ to $v$, but does not traverse the gadget from $a$ to $b$, or traverses the gadget from $a$ to $b$, but does not traverse the gadget from $u$ to $v$.
(This can be formally proven by noticing that every regular cycle visits all cities and additionally uses either the edge $x_1x_2$ or the edge $x_2x_3$.)
In order to simplify future notation, let us say that in the first case the cycle $C$ uses the edge $uv$, but does not use the edge $ab$,
while in the second case $C$ uses the edge $uv$, but does not use the edge $ab$ (these formulations are slightly imprecise since technically the edges $ab$ and $uv$ do not exist in $G_H$). 
In conclusion, every regular cycle has to visit the XOR-gadget and uses exactly one of the two edges $ab$ and $uv$. 

A XOR-gadget is called in \emph{semi-default state} if every of its four city gadgets is in semi-regular state (see \cref{fig:xor-gadget-semi-default}).
Observe that two XOR-gadgets can be applied to the same edge (see \cref{fig:multiple-xor-gadgets}).
\begin{figure}
    \centering
    \includegraphics[scale=1, page=1]{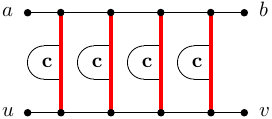}
    \caption{A XOR-gadget is called in semi-default state, if every of its city gadgets is in semi-default state.}
    \label{fig:xor-gadget-semi-default}
\end{figure}
%-----------------------------------------
\begin{figure}
    \centering
    \includegraphics[scale=1]{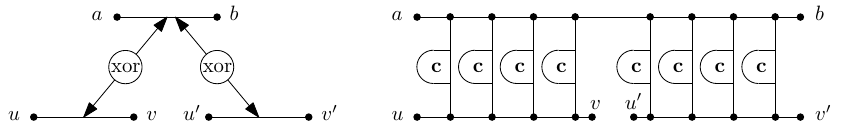}
    \caption{Multiple XOR-gadgets can be applied to the same edge.}
    \label{fig:multiple-xor-gadgets}
\end{figure}

Finally, we require the following technical lemma:

\begin{lemma}
\label{lem:XOR-bipartite}
    If $G$ is a bipartite graph and $e, e'$ are distinct edges, then one can add a XOR-gadget between $e$ and $e'$ such that the graph remains bipartite.
\end{lemma}
\begin{proof}
    Since every city gadget connects its two endpoints via an odd-length path, the XOR-gadget itself is bipartite, such that vertex set $\set{a,v}$ and vertex set $\set{u,b}$ lie in different parts of the bipartition.
    Let vertices $w_1, w_2$ be the endpoints of edge $e$ and vertices $w'_1, w'_2$ be the endpoints of edge $e'$ in the graph $G$.
    Since $G$ is bipartite, $w_1$ and $w_2$ lie in different parts, and $w'_1$ and $w'_2$ lie in different parts (w.r.t.\ the bipartition of $G$).
    We do a case distinction: 
    If the vertex sets $\set{w_1, w'_2}$ and $\set{w_2, w'_1}$ lie in different parts of the bipartition of $G$, 
    then we add the XOR-gadget in such a way that $(w_1,w_2)$ is identified with $(a,b)$ and $(w'_1, w'_2)$ is identified with $(u,v)$. The resulting graph is bipartite, since a bipartite graph is inserted into another bipartite graph respecting the bipartition.
    In the other case, if the vertex sets $\set{w_1, w'_1}$ and $\set{w_2, w'_2}$ lie in different parts of the bipartition of $G$,
    then we add the XOR-gadget in such a way that $(w_1,w_2)$ is still identified with $(a,b)$, but $(w'_1, w'_2)$ is identified with $(v,u)$. 
    This means that we effectively mirror the bottom edge of the XOR-gadget. The resulting graph is bipartite.
    
\end{proof}
\subsection{Ladders and the $\forall$-gadget}
\label{sec:ladders}

As explained in \cref{sec:technical-overview}, the $\forall$-gadget needs to exhibit asymmetric behavior.
In order to achieve this asymmetric behavior, we introduce yet another gadget, the \emph{ladder gadget}. 
We show that the ladder gadget exhibits a very slight asymmetric behavior, and amplify this effect by combining many ladder gadgets together.
A ladder gadget is an induced subgraph on the 14 vertices
\[
\set{a_0, \dots, a_6} \cup \set{b_0, \dots, b_6}
\]
and edge set
\[
\bigcup_{i=1}^5 \set{a_ib_i} \cup \bigcup_{i=0}^5 \set{a_ia_{i+1}, b_ib_{i+1}}.
\]
An example is depicted in \cref{fig:ladder-gadget:definition}. Ladder gadgets are similar to tower gadgets, but have a connection to both their top and bottom side. 
While tower gadgets have a variable height, a ladder gadget in this paper always has a fixed height of 5.
We will maintain the invariant that a ladder is connected to the rest of the graph $G_H$ only via the vertices $a_0, b_0, a_6, b_6$.
Consider a fixed PM $M$ of $G_H$. With respect to $M$ a ladder gadget is called in \emph{semi-default state}, if the vertices $a_0, b_0, a_6, b_6$ are matched to vertices outside the gadget by $M$.
It is in \emph{default state}, if $\fromto{a_1b_1}{a_5b_5} \subseteq M$. 
It is called the \emph{bottom-open ladder}, if $\set{a_5b_5, a_4a_3, b_4b_3, a_2a_1, b_2b_1} \subseteq M$. 
It is called the \emph{top-open ladder}, if $\set{a_5a_4, b_5b_4, a_3a_2, b_3b_2, a_1b_1} \subseteq M$ (see \cref{fig:ladder-gadget:default,fig:ladder-gadget:semi-default,fig:ladder-gadget:cap,fig:ladder-gadget:cup}).
%--------------------------------------------------------
%--------------------------------------------------------
\begin{figure}
    \centering
    \begin{subfigure}{0.15\textwidth}
        \includegraphics[scale=1,page=1]{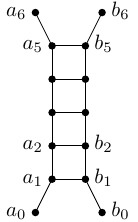}
        \caption{Definition}
    \label{fig:ladder-gadget:definition}
    \end{subfigure}
    \hfill
    \begin{subfigure}{0.18\textwidth}
        \includegraphics[scale=1,page=2]{img/ladder-gadget.pdf}
        \caption{semi-default}
    \label{fig:ladder-gadget:semi-default}
    \end{subfigure}
    \hfill
    \begin{subfigure}{0.15\textwidth}
        \includegraphics[scale=1,page=3]{img/ladder-gadget.pdf}
        \caption{default}
    \label{fig:ladder-gadget:default}
    \end{subfigure}
    \hfill
    \begin{subfigure}{0.19\textwidth}
    \centering
        \includegraphics[scale=1,page=4]{img/ladder-gadget.pdf}
        \caption{bottom-open}
    \label{fig:ladder-gadget:cap}
    \end{subfigure}
    \hfill
    \begin{subfigure}{0.15\textwidth}
        \includegraphics[scale=1,page=5]{img/ladder-gadget.pdf}
        \caption{top-open}
    \label{fig:ladder-gadget:cup}
    \end{subfigure}
    \hfill
    \caption{A ladder gadget. Ladder gadgets always have a fixed height of 5.}
    \label{fig:ladder-gadget}
\end{figure}

Similar to the case of tower gagdets, we again want to consider well-behaved cycles. 
Let $L$ be some ladder gagdet. Assume we have some (not necessarily alternating) cycle $C$ in the graph $G_H$, such that $C$ includes at least one vertex of $L$. 
The cycle $C$ is called \emph{well-behaved} for the ladder gadget $L$, if either
\begin{itemize}
    \item both $a_0a_1, b_0b_1 \in E(C)$ and both $a_5a_6, b_5b_6 \not\in E(C)$, or
    \item both $a_5a_6, b_5b_6 \in E(C)$ and both $a_0a_1, b_0b_1 \not\in E(C)$.
\end{itemize}
In other words, the cycle $C$ enters the ladder either from the top (i.e.\ via the edge $a5a_6$ or $b_5b_6$) or from the bottom (i.e.\ via the edge $a_0a_1$ or $b_0b_1$), 
leaves from the the same side it entered, and after that does not visit the inner ladder (vertices $a_i, b_i$ for $i \in \fromto{1}{5}$) again.
In the first case, let us say that the well-behaved cycle $C$ visits $L$ \emph{from the top}, in the second case, let us say that $C$ visits $L$ \emph{from the bottom}.
The next lemma establishes the already claimed asymmetry, by showing that in order to transform the bottom-open ladder 
(the top-open ladder, respectively) into the default ladder with well-behaved cycles, we need to primarily come from the bottom (the top, respectively). 
Similarly to before, a sequence $(C_1,\dots,C_d)$ of cycles is called a \emph{well-behaved flip sequence} for the ladder $L$ that transforms $M_1$ into $M_2$, if each of $C_1,\dots,C_d$ is well-behaved, and restricted to $L$, we have that $C_i$ is an alternating path with respect to $M \symdiff C_1 \symdiff \dots \symdiff C_{i-1}$ for all $i \in [d]$, and $M_2 = M_1 \symdiff C_1 \symdiff \dots \symdiff C_d$.
\cref{fig:ladder-transform} showcases a well-behaved flip sequence of length 2.
\begin{figure}
    \centering
        \includegraphics[scale=1,page=1]{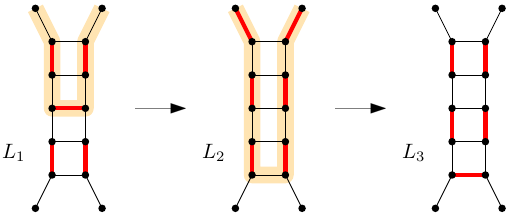}
    \caption{Ladder $L_1$ is transformed into ladder $L_3$ with two well-behaved alternating cycles from the top. Ladders $L_1$ and $L_3$ are in semi-default state, ladder $L_2$ is not.}
    \label{fig:ladder-transform}
\end{figure}

\begin{lemma}
\label{lem:ladder-necessary}
    Let $L$ be a ladder gadget. 
    If $(C_1, \dots, C_d)$ is a well-behaved flip sequence for $L$ that transforms the bottom-open ladder into the default state, then $d \geq 4$.
    Additionally if $d = 4$, then each of $C_1,\dots,C_4$ must come from the bottom.
    The analogous statement is true if $L$ starts as the top-open ladder, except that now each of $C_1,\dots,C_4$ must come from the top.
\end{lemma}
\begin{proof}
    Observe that every well-behaved cycle $C$ includes exactly one edge of the form $a_ib_i$, where $i \in [5]$. 
    (This follows in particular since a well-behaved cycle visits a ladder only once, either from the top or from the bottom.)
    Let $\mathcal{H}(M) := \set{i \in [5] : a_ib_i \in M}$ be the index set of horizontal edges that some PM $M$ contains. 
    The previous fact implies that after flipping a well-behaved cycle, the set $\mathcal{H}(M)$ changes by only addition or deletion of a single element.
    Since for the default ladder $\mathcal{H}(M) = \fromto{1}{5}$, and for the bottom-open ladder $\mathcal{H}(M') = \set{5}$, 
    we see that if $(C_1,\dots,C_d)$ is a flip sequence transforming the former into the latter, then $d \geq 4$.
    Further, if $d = 4$, 
    assume one of the cycles $C_j$ for $j \in [4]$ comes from the top. 
    Then either the cycle $C_j$ or some cycle $C_{j'}$ for $j' < j$ contains the edge $a_5b_5$ as the only horizontal edge.
    This is a contradiction to the fact that we go from $\mathcal{H}(M) = \fromto{1}{5}$ to $\mathcal{H}(M') = \set{5}$ in only four steps. 
    Finally, the statement about the top-open ladder holds analogously by vertical symmetry.
\end{proof}

We remark that in the above lemma it is crucial that the cycles are well-behaved. 
In particular, if a cycle is allowed to enter a ladder from the top, but leave from the bottom, 
then one can show that it is possible to transform the bottom-open ladder into the default ladder with only two alternating cycles.
Our construction of $G_H$ will ensure that in a short enough flip sequence, the majority of all ladders only ever interact with well-behaved cycles.

The necessary condition of \cref{lem:ladder-necessary} for transforming a ladder into another 
is complemented by the following sufficient condition.
\begin{lemma}
\label{lem:ladder:sufficient}
    Let $L$ be a ladder gadget. Let $M_1$ and $M_2$ be two semi-default states of $L$. 
    There is a well-behaved flip sequence $(C_1,\dots,C_4)$ of length 4 that transforms $M_1$ into $M_2$. 
    Furthermore there is one such sequence so that $C_1,C_2$ both come from the top, or both come from the bottom, and $C_3,C_4$ both come from the top, or both come from the bottom.
\end{lemma}
\begin{proof}
    Observe that there are 8 different semi-default ladders, since a semi-default ladder has either $1,3$ or $5$ horizontal edges, 
    and all possibilities to arrange these are displayed in \cref{fig:ladder-graph}.
    Furthermore, \cref{fig:ladder-graph} shows a graph, where two ladders are connected with an edge, 
    if one can be transformed into the other using two well-behaved cycles, either both from the top (denoted $2t$), or both from the bottom (denoted $2b$).
    Note that this relation is symmetric, i.e.\ if $L_1$ can be transformed into $L_2$ with two well-behaved cycles, then $L_2$ can also be transformed into $L_1$.
    To prove the lemma it now suffices to observe that the graph of \cref{fig:ladder-graph} has diameter 2, i.e.\ every pair of ladders is connected by a path with at most two edges. 
    Finally, we remark that if $(C_1,\dots,C_d)$ with $d = 2$ or $d=0$ is an even shorter flip sequence from $M_1$ to $M_2$, 
    we can also find a flip sequence of length exactly four: 
    We can simply choose any well-behaved cycle $C$ from the top or bottom (at least one always exists), 
    and consider the sequence $(C_1,\dots,C_d, C, C)$. 
\end{proof}
\begin{figure}
    \centering
        \includegraphics[scale=1,page=1]{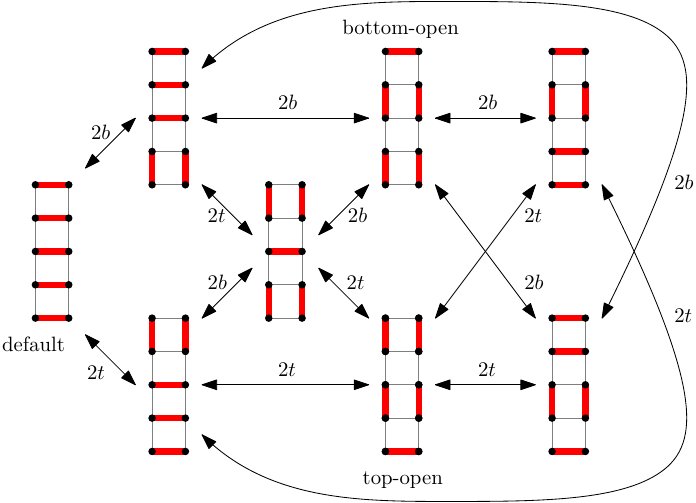}
    \caption{All 8 semi-default ladders. Arrows indicate that one can be turned into the other by using two well-behaved cycles, either both from the top ($2t$) or both from the bottom $(2b)$. In order to transform the top-open ladder into the default ladder, one requires at least 4 well-behaved cycles, all from the top.}
    \label{fig:ladder-graph}
\end{figure}

\textbf{$\forall$-gadgets}. After we showed that a single ladder gadget can exhibit a slight asymmetry between top and bottom, 
we now show how to amplify this asymmetry by combining a lot of ladders into a single gadget.
Assume we are given an undirected graph $G_H$, and three distinct vertices called $v_\text{out}, u_\text{in}, w_\text{in}$.  
Let $t \in \N$. A \emph{$\forall$-gadget} of \emph{width} $t$ is the subgraph that is depicted in \cref{fig:forall-gadget-definition} such that it is connected to the rest of $G_H$ via only these three vertices.
%---------------------------------
\begin{figure}
    \centering
        \includegraphics[scale=0.8,page=1]{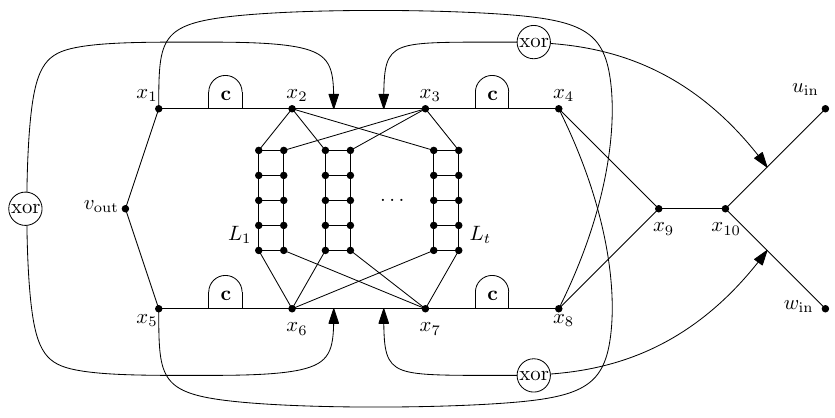}
    \caption{A $\forall$-gadget between the three vertices $v_\text{out}, u_\text{in}$, and $w_\text{in}$.}
    \label{fig:forall-gadget-definition}
\end{figure}
%---------------------------------
Precisely, we introduce 10 new vertices $\set{x_1,\dots,x_{10}}$. 
After that, we connect the four pairs $(x_1,x_2), (x_3,x_4), (x_5,x_6), (x_7,x_8)$ with a city gadget each.
We add the edge set
\[
\set{v_\vout x_1, v_\vout x_5, x_1x_8, x_5x_4, x_2x_3, x_6x_7, x_4x_9, x_8x_9, x_9x_{10}, x_{10}u_\vin, x_{10}w_\vin}.
\]
Three XOR-gadgets are applied to the edge pairs $(x_2x_3, x_6x_7)$, $(x_2x_3, x_{10}u_\vin)$, and $(x_6x_7, x_{10}w_\vin)$ respectively.
Finally, a total amount of $t$ ladder gadgets are added.
For each of these $t$ ladders, vertex $a_6$ in the ladder is identified with vertex $x_2$ in the $\forall$-gadget. 
Likewise, $b_6$ is identified with $x_3$, $a_0$ is identified with $x_6$, and $b_0$ is identified with $x_7$.
This completes the description of the $\forall$-gadget.

A $\forall$-gadget is called in \emph{semi-default state} with respect to a fixed PM $M$, if all its city-gadgets, and all its XOR-gagdets are in semi-default state, and $x_9x_{10} \in M$ (see \cref{fig:forall-gadget-semi-default}).
\begin{figure}
    \centering
        \includegraphics[scale=0.6,page=2]{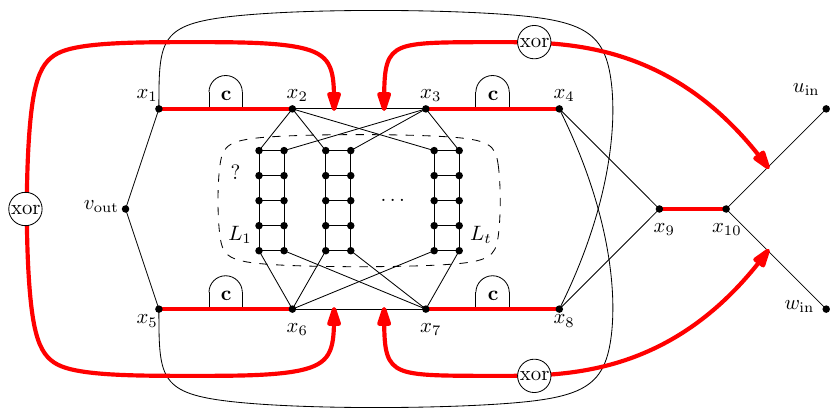}        
    \caption{A $\forall$-gadget is in semi-default state, if all its city gadgets and all its XOR-gadgets are in semi-default state and the edge $x_9x_{10}$ is matched.}
    \label{fig:forall-gadget-semi-default}
\end{figure}
%---------------------
We now describe the main functionality of a $\forall$-gadget.
Recall that a regular cycle is a (not necessarily alternating) cycle in $G_H$ that visits all cities.
Let $A$ be a $\forall$-gadget and $C$ be a regular cycle.
We say that the cycle $C$ is in \emph{top state} with respect to the gadget $A$, if $C$ uses the edge $x_{10}u_\vin$. 
Analogously, we say that the cycle $C$ is in \emph{bottom state} with respect to the gadget $A$, if $C$ uses the edge $x_{10}w_\vin$.
Let us for the remainder of this section assume w.l.o.g.\ that the graph $G_H$ contains at least one city gadget outside of $A$, such that every regular cycle needs to both enter and leave $A$.
Let us say that $C$ \emph{visits} some ladder, if $V(C) \cap \left( \bigcup_{i=1}^5 \set{a_i, b_i} \right) \neq \emptyset$ for the vertices $a_0,\dots, a_6, b_0, \dots, b_6$ of that ladder.
%----------------------------------------------
\begin{figure}
    \begin{subfigure}{1\textwidth}
        \centering
        \includegraphics[scale=0.8,page=4]{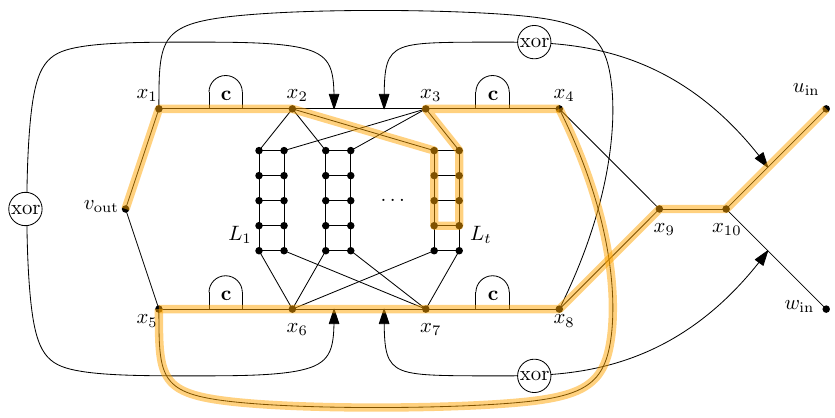}
        \caption{Top state}
    \label{fig:forall-gadget-states:top}
    \end{subfigure}
    \begin{subfigure}{1\textwidth}
        \centering
        \includegraphics[scale=0.8,page=3]{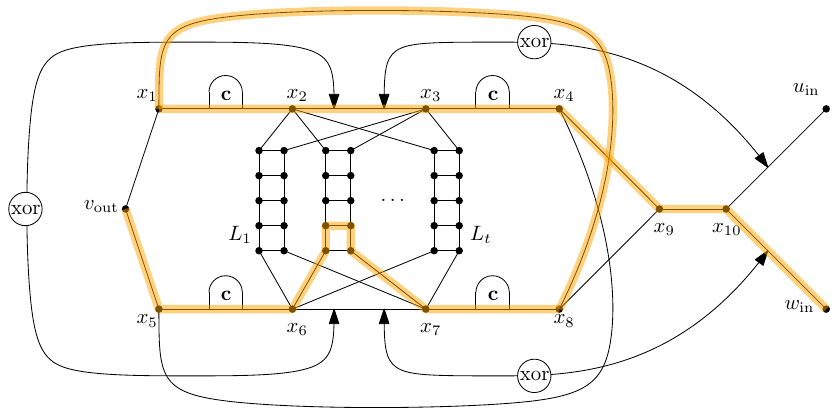}
        \caption{Bottom state}
    \label{fig:forall-gadget-states:bottom}
    \end{subfigure}
    \caption{If $C$ is a cycle visiting every city gadget, then a $\forall$-gadget is in one of the following two states: Either $C$ contains the edge $x_{10}u_\text{in}$ and visits exactly one ladder coming from the top (\emph{top state}), or $C$ contains the edge $x_{10}w_\text{in}$ and visits exactly one ladder coming from the bottom (\emph{bottom state}).}
    \label{fig:forall-gadget-states}
\end{figure}
%----------------------------------------------
\begin{lemma}
\label{lem:forall-gagdet}
    If $A$ is a $\forall$-gadget, and $C$ is a regular cycle in $G_H$, then $C$ is either in top state or bottom state (but not both).
    Furthermore, the cycle $C$ visits exactly one ladder $L$ of the gadget, and $C$ is well-behaved for $L$. 
    If the cycle $C$ is in top state (bottom state, respectively) then it visits $L$ from the top (from the bottom, respectively).
\end{lemma}
\begin{proof}
    Since $C$ is regular, the city-gadgets and XOR-gadgets behave as intended. In particular, 
    the cycle $C$ uses all four edges $x_1x_2, x_3x_4, x_5x_6, x_7x_8$, and exactly one of the two edges $x_{10}u_\vin$ and $x_{10}w_\vin$ (because they are connected by a chain of XOR-gadgets).
    Since we assumed that there is at least one city gadget outside of the $\forall$-gadget $A$, 
    and $A$ is only connected to the rest of the graph via the three vertices $v_\vout, u_\vin, w_\vin$, the cycle $C$ enters the gadget in vertex $v_\vout$ and leaves via either vertex $u_\vin$ or $w_\vin$.
    We now distinguish between the following cases.

    \textbf{Case 1:} $C$ uses $x_{10}u_\vin$. In this case $C$ also uses $x_6x_7$, but does not use the two edges $x_2x_3, x_{10}u_\vin$. 
    Since $C$ enters the gadget at $v_\vout$, it uses exactly one of the two edges $v_\vout x_1$ and $v_\vout x_5$.

    \textbf{Case 1a:} $C$ uses $v_\vout x_1$. In this case, consider vertex $x_5$. Since it has degree 3, and we already know that $C$ uses $x_5x_6$, but not $v_\vout x_5$, we deduce that $C$ uses $x_4x_5$.
    Then, combining all the information we have so far, we see that the cycle $C$ looks like in \cref{fig:forall-gadget-states:top}. Precisely, $C$ enters at $v_\vout$, goes to $x_1$, traverses from $x_1$ to $x_4$ by entering a ladder gadget, goes to $x_5$, and then visits vertices $x_6$ - $x_{10}$ in that order before leaving via $u_\vin$.
    In particular, the cycle $C$ visits exactly one ladder. This is because in order to enter or exit a ladder, the cycle needs to visit one of the four vertices $x_2, x_3, x_6$ or $x_7$. 
    However, $x_6$ and $x_7$ are blocked.
    Therefore $x_2$ needs to be used to enter a ladder, and $x_3$ needs to be used to exit the same ladder, and no different ladder can be entered after that.
    Hence $C$ visits exactly one ladder, coming from the top.

    \textbf{Case 1b:} $C$ uses $v_\vout x_5$. In that case, the argument is analogous to case 1a, with the exception that $C$ uses the edge $x_1x_8$. 
    Still, it remains that $x_6, x_7$ are blocked, so $C$ visits exactly one ladder coming from the top.
    To summarize, in both cases 1a and 1b, the cycle visits a ladder from the top and leaves the gadget via the edge $x_{10}u_\vin$.

    \textbf{Case 2:} $C$ uses $x_{10}w_\vin$. Then, due to vertical symmetry, the cycle visits exactly one ladder from the bottom, and leaves via the edge $x_{10}w_\vin$.
\end{proof}

The above lemma shows that every regular cycle can only interact with one ladder in a well-behaved manner. But what about the irregular cycles? 
Since irregular cycles do not need to visit every city, they can potentially interact with our gadgets in a much less controlled way. Still, the \enquote{damage} done by an irregular cycle is restricted, as the following lemma shows.

\begin{lemma}
\label{lem:damage-irregular-cycle}
    If $A$ is a $\forall$-gagdet, and $C$ is some cycle in $G_H$, then $C$ visits at most 4 ladders inside $A$.
\end{lemma}
\begin{proof}
    Observe that in the graph $G_H - \set{x_2, x_3, x_6, x_7}$, each of the $t$ ladders of $A$ is its own connected component.
    Whenever the cycle $C$ visits a ladder, it needs to enter and exit using one of these four vertices.
    Hence every ladder that $C$ visits is associated to two vertices of $x_2, x_3, x_6, x_7$. 
    (We assume that $C$ enters and exists the ladder, because if $C$ is entirely contained in it, we are done).
    On the other hand, each of these four vertices is associated to at most two ladders.
    Hence $C$ visits at most 4 ladders inside $A$.
\end{proof}

\subsection{Definition of $G_H$}
\label{sec:combining-the-pieces}

We are now ready to describe the main construction. Assume we are given an instance of $\forall \exists$-\textsc{HamCycle}. This instance consists out of a directed graph $H$ together with some $k \in \N$ and vertices $v^{(1)}, \dots, v^{(k)}$ in $H$, each having outdegree 2. Let $n := |V(H)|$.
For $i \in [k]$, let $u^{(i)}$ and $w^{(i)}$ denote the two outneighbors of vertex $v^{(i)}$. Define the two arcs $e_i := (v^{(i)}, u^{(i)})$ and $\overline e_i := (v^{(i)}, w^{(i)})$, 
as well as the set $E' := \set{e_1, \overline e_1, \dots, e_k, \overline e_k}$.
The question is whether for all patterns $P \subseteq E'$ the directed graph $H$ has a Hamiltonian cycle respecting $P$.
Given this instance, the undirected graph $G_H$ is defined in two steps. In a first step, we consider the set of $2n$ vertices
\[
\bigcup_{x \in V(H)} \set{x_\vin, x_\vout}.
\]
In a second step, we add the following edges and gadgets:
\begin{itemize}
    \item For $x \in V(H)$, we connect $x_\vin$ to $x_\vout$ with a city gadget.
    \item For $i \in [k]$, we connect the three vertices $v^{(i)}_\vout, u^{(i)}_\vin$ and $w^{(i)}_\vin$ with a $\forall$-gadget.
    \item For each arc $(x,y) \in E(H) \setminus E'$, we add the edge $x_\vout y_\vin$ to $G_H$.
\end{itemize}
Furthermore, we let all city gadgets, including those inside the XOR- and $\forall$-gadgets have a height of $2n^4$ and width of $2n^4 + 100n$. We let all $\forall$-gadgets have a width of $t = n^4$ (i.e.\ they contain exactly $n^4$ ladders each).
This completes the description of $G_H$. A schematic of the construction is depicted in \cref{fig:full-reduction}.
Observe that $G_H$ contains a total of $k$ $\forall$-gadgets, $3k$ XOR-gagdets, and $n + 16k \leq 17n$ city gadgets. Since each city gadget has $\Theta(n^8)$ vertices, and these are the majority of all vertices, we have $|V(G_H)| = \Theta(n^9)$.
In particular $G_H$ can be constructed in polynomial time given $H$.
\begin{figure}
    \centering
        \includegraphics[scale=0.8]{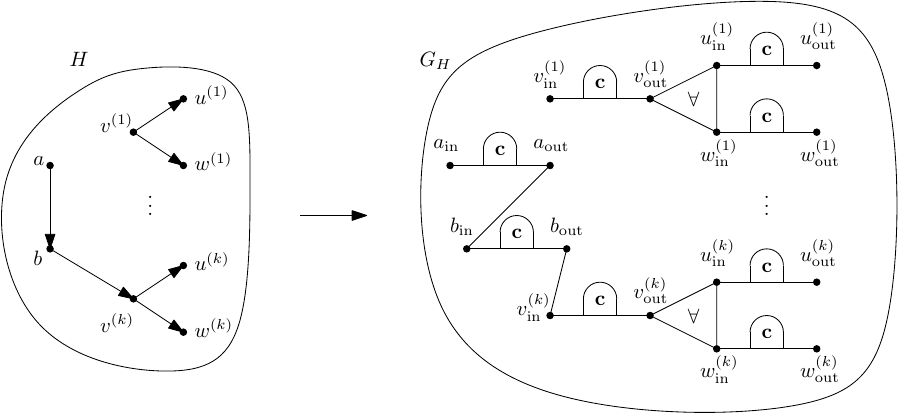}
    \caption{Construction of the graph $G_H$ from $H$.}
    \label{fig:full-reduction}
\end{figure}

\begin{lemma}
The graph $G_H$ is bipartite.
\end{lemma}
\begin{proof}
    First consider the graph $G'_H$ obtained from $G_H$ by deleting all the $\forall$-gadgets with exception of the vertices $v^{(i)}_\vout, u^{(i)}_\vin, w^{(i)}_\vin$.
    The graph $G'_H$ is clearly bipartite, since every edge and every city gadget connects some vertex $x_\vout$ to some vertex $y_\vin$, and the city gadgets connect their endpoints with an odd-length path.
    Consider \cref{fig:forall-gadget-definition}. 
    By assigning for $i \in [10]$ the vertex $x_i$ to color class $i \bmod 2$, and using \cref{lem:XOR-bipartite}, we see that a $\forall$-gadget is bipartite, such that the two sets $\set{v_\vout}$ and $\set{u_\vin, w_\vin}$ are contained in different parts of the bipartition.
    Since in $G'_H$ the three vertices $v^{(i)}_\vout, u^{(i)}_\vin, w^{(i)}_\vin$ have the analogous behavior with respect to the bipartition of $G'_H$, we see that if we add the $\forall$-gadgets back into $G'_H$, the graph remains bipartite. Hence $G_H$ is bipartite.
\end{proof}

For the proof of our main theorem, we define a particular PM of $G_H$, called the \emph{default PM} $M_\text{def}$.  The PM $M_\text{def} \subseteq E(G_H)$ is defined by having
\begin{itemize}
    \item all towers of all city gadgets in $G_H$ in default state
    \item all $\forall$-gagdets in $G_H$ in semi-default state
    \item all $t$ ladders in every $\forall$-gagdet in default state.
\end{itemize}
(By the term \enquote{all city gadgets in $G_H$}, we mean all $n + 16k$ gadgets, including those inside other gadgets.) Observe that this is a well-defined perfect matching of $G_H$, since every vertex $x_\vin, x_\vout$ for $x \in V(H)$, as well as every vertex inside a city gagdet or $\forall$-gagdet is touched exactly once.
Furthemore, let us say that some PM $M'$ of $G_H$ is in \emph{semi-default} state, if with respect to $M'$, all the city gadgets in $G_H$ and all $\forall$-gadgets in $G_H$ are in semi-default state.
The following shows that the semi-default PMs lie in a certain sense \enquote{dense} in the 1-skeleton of $P_{G_H}$.

\begin{lemma}
\label{lem:semi-default-dense}
    For every PM $M$ of $G_H$, there exists some PM $M'$ of $G_H$ in semi-default state, with flip distance $\dist(M, M') \leq 23n$.
\end{lemma}
\begin{proof}
    Consider the vertex set $V_s \subseteq V(G_H)$ defined by
    \[
    V_s := \bigcup_{x \in V(H)}\set{x_\vin} \cup \bigcup_X \set{x_1^{(X)}, \dots, x_4^{(X)}} \cup \bigcup_A \set{x^{(A)}_1, \dots, x^{(A)}_{10}}. 
    \]
    Here, the index $X$ runs over all XOR-gadgets in $G_H$, the index $A$ runs over all $\forall$-gadgets in $G_H$, and $x_i^{(X)}$ denotes the corresponding vertex in the XOR-gadget 
    (see \cref{fig:xor-gadget}) and $x_i^{(A)}$ denotes the corresponding vertex in the $\forall$-gadget (see \cref{fig:forall-gadget-definition}).
    We make the following claim: A PM $M'$ is in semi-default state if and only if $M'$ assigns the same partner as the PM $M_\text{def}$ to every vertex $v \in V_s$.
    Indeed, if we recall the definitions of a semi-default state of a city gadget, a XOR-gadget, and a $\forall$-gadget, 
    we see that whenever for some PM $M'$ the vertices of $V_s$ are matched in the same way as in $M_\text{def}$, the PM $M'$ must be in semi-default state (we use that every tower has an even number of vertices).
    We have $|V_s| = n + 12k + 10k \leq 23n$.
    Now, given an arbitrary PM $M$, consider the symmetric difference $M \symdiff M_\text{def}$. 
    Since $|V(G_H)| = \Theta(n^9)$, the set $M \symdiff M_\text{def}$ could in general contain up to $O(n^9)$ vertex-disjoint cycles. 
    Let $C_1, \dots, C_t$ be those cycles of $M \symdiff M_\text{def}$ that contain some vertex of $V_s$. Then $t \leq 23n$.
    If we flip the vertex-disjoint cycles $C_1, \dots, C_t$ one after another, we transform $M$ into a semi-default PM $M'$, because $M'$ assigns the same partners as $M_\text{def}$ to vertices in $V_s$. By construction $\dist(M, M') \leq 23n$.    
\end{proof}

\subsection{Proof of the main theorem}
\label{sec:proof-of-main-thm}

In this subsection we prove \cref{lem:main-thm-if,lem:main-thm-only-if}, which together prove our main theorem. 
Consider throughout this subsection a fixed instance of $\forall \exists$-\textsc{HamCycle} consisting out of a directed graph $H$ with vertices $v^{(1)},\dots, v^{(k)}$ 
and edges $E' = \set{e_1, \overline e_1, \dots, e_k, \overline e_k}$. Let $G_H$ be defined as in \cref{sec:combining-the-pieces}.
As is standard in hardness reductions, we can assume a lower bound on the size of the input instance. 
Specifically, we assume w.l.o.g.\ that $n = |V(H)| \geq 83\, 000$.

\mainTheoremIf*
\begin{proof}
    
    For all patterns $P \subseteq E'$, let $C_P$ be a Hamiltonian cycle respecting $P$ in the directed graph $H$. 
    Let $M_1, M_2$ be two arbitrary PMs of $G_H$. We show that $\dist(M_1, M_2) \leq 4n^4 + 46n$. Since $M_1, M_2$ are arbitrary, this proves the lemma.
    First, due to \cref{lem:semi-default-dense}, there exist PMs $M'_1, M'_2$ of $G_H$ in semi-default state and $\dist(M_1, M'_1) \leq 23n$ and $\dist(M_2, M'_2) \leq 23n$. 
    We claim $\dist(M'_1, M'_2) \leq 4n^4$. If we can show this claim, we are done, since then
    \[
    \dist(M_1, M_2) \leq \dist(M_1, M'_1) + \dist(M'_1, M'_2) + \dist(M'_2, M_2) \leq 4n^4 + 46n.
    \]

    We give a quick overview of the proof. We have the task to transform $M'_1$ into $M'_2$ using a short flip sequence $(C_1,\dots,C_d)$ with $d \leq 4n^4$.
    Assume we can find a flip sequence with the property that it correctly transforms $M'_1$ into $M'_2$ locally restricted to each city gadget, 
    and correctly transforms $M'_1$ into $M'_2$ locally restricted to each $\forall$-gadget.
    Then we are done, since both $M'_1,M'_2$ are semi-default and so their difference is contained in the union of all city gadgets and $\forall$-gadgets.
    In order to correctly perform the transformation restricted to the city gadgets, we will ensure that each of $C_1,\dots,C_d$ is a regular cycle. 
    Then we can visit each city and transform each tower of height $h$ with $2h = 2n^4$ cycles.
    In order to correctly perform the transformation restricted to the $\forall$-gadgets, 
    we note that the main difficulty is the correct transformation of the ladders.
    Due to \cref{lem:ladder:sufficient}, we can transform any semi-default ladder into any other semi-default ladder in only 4 steps, if we can control the direction we are coming from.
    Indeed, the main idea of the proof is that, since $H$ contains a Hamiltonian cycle for \emph{any} pattern, 
    we can for each of the $k$ $\forall$-gadgets $A_1,\dots,A_k$ choose independently, whether we want to traverse it in top state or bottom state.
    Furthermore, this choice does not need to be fixed, but can be changed arbitrarily many times during the flip sequence $(C_1,\dots,C_d)$ independently for each $A_j$, where $j \in [k]$.

    We are now ready to start the proof. Let $T$ be a fixed tower gadget in $G_H$. Since $T$ has height $2n^4$ and $T$ is in semi-default state with respect to both $M'_1$ and $M'_2$, due to \cref{lem:tower:upper-bound} there exists a well-behaved flip sequence $(P_1^{(T)},\dots, P_d^{(T)})$ for $T$ of length $d = 4n^4$ that transforms the PM $M'_1$ restricted to $T$ into the PM $M'_2$ restricted to $T$.
    
    Let $A_j$ be a fixed $\forall$-gadget of $G_H$ for some $j \in [k]$. Note that $A_j$ contains $n^4$ ladders $L_1, \dots, L_{n^4}$. 
    Let $i \in [n^4]$ and consider the ladder $L_i$. Since both $M'_1$ and $M'_2$ are in semi-default state, 
    the ladder $L_i$ is semi-default with respect to both $M'_1$ and $M'_2$.
    Due to \cref{lem:ladder:sufficient} there exists a sequence $(P_1^{(L_i)}, \dots, P_4^{(L_i)})$ of length 4 of well-behaved paths for $L_i$ 
    transforming the PM $M'_1$ restricted to $L_i$ into the PM $M'_2$ restricted to $L_i$, and such that both $P_1^{(L_i)},P_2^{(L_i)}$ come from the same direction (top/bottom) 
    and both $P_3^{(L_i)},P_4^{(L_i)}$ come from the same direction (top/bottom).
    We can represent this fact as a string $s(L_i) \in \set{tt, bb}^2$, i.e.\ one of the strings $tttt, ttbb, bbtt, bbbb$. More formally, the string $s(L_i)$ has four characters, indicating that if one can visit the ladder $L_i$ with four cycles coming from either the top ($t$) or bottom ($b$) in the same sequence as in the string $s(L_i)$, then one can transform $M'_1$ into $M'_2$ restricted to the ladder $L_i$.
    We can concatenate the strings $s(L_i)$ for $i = 1$ up to $i = n^4$ and obtain a string
    \[
    s(A) \in \set{tt, bb}^{2n^4}.
    \]
    We can interpret the string $s(A_j)$ as the \enquote{total demand} of the $\forall$-gadget $A_j$. 
    If we can manage to visit the ladders of $A_j$ with $4n^4$ well-behaved cycles from the top or bottom in the same sequence as the string $s(A_j)$ describes, we can transform $M'_1$ into $M'_2$ (restricted to $A_j$).
     Let $A_1, \dots, A_k$ be the $k$ distinct $\forall$-gadgets contained in $G_H$. Observe that the demand strings $s(A_j), s(A_{j'})$ can in general be very different from each other for distinct $j,j' \in [k]$.
    However, we now crucially use the property that $H$ has a Hamiltonian cycle for \emph{all} patterns $P \subseteq E'$.
    Concretely, for all $j \in [k]$ consider the first two characters of the string $s(A_j)$, denoted by $p_j \in \set{tt, bb}$. 
    We can define a pattern $P \subseteq E' \subseteq E(H)$ based on these characters in a natural way by letting
    \[
    P := \set{e_j : j \in [k], p_j = tt} \cup \set{\overline e_j : j \in [k], p_j = bb}.
    \]
    
    We define a cycle $C$ in $G_H$ depending on those characters as follows:
    \begin{itemize}
        \item For all $j \in [k]$, if $p_j = tt$, then $C$ traverses the $\forall$-gadget $A_j$ in top state and visits ladder $L_1$ in that gadget from the top.
        \item Otherwise, if $p_j = bb$, then $C$ traverses the $\forall$-gadget $A_j$ in bottom state and visits ladder $L_1$ in that gadget from the bottom.
        \item $C$ visits all city gadgets in $G_H$
        \item In between the cities and $\forall$-gadgets, the cycle $C$ follows globally the same path as the Hamiltonian cycle $C_P$. 
        Formally, for all arcs $(a,b) \in E(C_P) \setminus E'$, we have $a_\vout b_\vin \in E(C)$.
        \item Restricted to each ladder gadget $L$, the cycle $C$ is equal to $P^{(L)}_1$. Restricted to each tower gadget $T$, the cycle $C$ is equal to $P^{(T)}_1$.
    \end{itemize}
    
    \textbf{Claim 1:} $C$ is a well-defined cycle in $G_H$ and is alternating with respect to $M'_1$.
    
    \textit{Proof of the claim.} We split $C$ into those segments inside and those segments outside some $\forall$-gadget.
    Restricted to some $\forall$-gadget $A$, the cycle $C$ is a well-defined path, since both top-state and bottom-state have this property (see \cref{fig:forall-gadget-states}). 
    Furthermore, inside of $A$ the cycle $C$ is alternating with respect to $M'_1$, since $M'_1$ is in semi-default state (see \cref{fig:forall-gadget-semi-default}).
    In particular, note that since the XOR-gadgets are in semi-default state, 
    both ways in which $C$ could interact with the XOR-gadget are alternating (see \cref{fig:xor-gadget-semi-default}).  
    Restricted to the ladder $L_1$ inside of $A$, the cycle $C$ is equal to $P^{(L_1)}_1$ and hence $M'_1$-alternating as well.
    Observe that for all $j \in [k]$ if $p_j = tt$, then $C$ traverses gadget $A_j$ in top state, i.e.\ it goes from $v^{(j)}_\vout$ to $u^{(j)}_\vin$. 
    We can interpret this as corresponding to the arc $e_j = (v^{(j)}, u^{(j)}) \in E(H)$.
    Likewise, if $p_j = bb$, we can interpret this as the arc $\overline e_j \in E(H)$.

    Outside the $\forall$-gadgets, the cycle $C$ is a well-defined cycle, because whenever $C$ encounters some city gadget of some vertex $x \in V(H)$, it uses the gadget to go from $x_\vin$ to $x_\vout$.
    In order to change from one city gadget to another, $C$ uses either a $\forall$-gadget or some edges corresponding to arcs of $C_P \subseteq H$.
    If $C$ uses a $\forall$-gadget it behaves just like the pattern $P$.
    Since $C_P$ is a Hamiltonian cycle of $H$ respecting the pattern $P$, we conclude that $C$ is a single cycle in $G_H$ that visits every city gadget. 
    Since every city gadget is in semi-default state in $M'_1$, and since for every tower $T$ the cycle $C$ is equal to $P^{(T)}_1$, the cycle $C$ is $M'_1$-alternating. 
    %(More precisely, note that for the transition points between the $\forall$-gadgets and the rest of the graph, $C$ is also alternating, because it both enters and leaves the $\forall$-gadget $A$ via an unmatched edge 
    %and immediately traverses the city gadgets corresponding to $v$ and $u$/$w$.)
    This proves the claim.
    
    We define a second cycle $C'$ in $G_H$, by specifying that $C'$ has exactly the same edges as $C$, with the exception of all tower gadgets and all ladders that $C$ traverses. 
    For all towers $T$, the cycle $C'$ is now equal to $P^{(T)}_2$ instead of $P^{(T)}_1$. 
    Analogously, for all ladders $L$ that $C$ visits, the new cycle $C'$ is equal to $P^{(L)}_2$ instead of $P^{(L)}_1$.

    \textbf{Claim 2:} $C'$ is a well-defined cycle in $G_H$ and is alternating with respect to $M'_1 \symdiff C$, and $M'_1 \symdiff C \symdiff C'$ is a PM in semi-default state.
    
    \textit{Proof of the claim.} Since $C$ is a well-defined cycle and $C'$ differs from it only on the towers and ladders, $C'$ is also a well-defined cycle (note that if $C$ and $C'$ differ on a ladder, both come from the same direction).
    Consider all edges of $C'$ not part of some ladder or tower. 
    All these edges are also part of $C$, and $C$ is $M'_1$-alternating.
    Hence $C'$ is alternating with respect to $M'_1 \symdiff C$ on those edges.
    Inside some tower $T$, the cycle $C'$ is also alternating with respect to $M'_1 \symdiff C$, because $(P_1^{(T)}, P^{(T)}_2, \dots )$ is a well-behaved flip sequence. The same holds for ladders.
    Finally, $M'_1 \symdiff C \symdiff C'$ is in semi-default state, because $M'_1$ is in semi-default state, 
    and since both  the cycles $C, C'$ traverse all city gadgets of the graph (as well as the edge $x_9x_{10}$ of every $\forall$-gadget). 
    Hence we can check that in every city gadget, every XOR-gadget, and every $\forall$-gadget the PM $M'_1 \symdiff C \symdiff C'$ is again in semi-default state 
    (because, intuitively speaking, every city gadget and the edges $x_9x_{10}$ were flipped twice).
    This proves the claim.

    To summarize, we have found two cycles $C, C'$ such that by flipping $C,C'$ we have eliminated the first two characters of all demand strings $s(A_j)$, 
    for $j \in [k]$ (which have length $4n^4$) and for each tower $T$ somewhere in $G_H$ we have successfully flipped the first two paths $P^{(T)}_1, P^{(T)}_2$ of its flip sequence (which has length $4n^4$). 
    Furthermore the PM $M'_1 \symdiff C \symdiff C'$ is again in semi-default state.
    Iterating this argument shows that one can reach $M'_2$ from $M'_1$ by flipping $4n^4$ cycles, hence $\dist(M'_1, M'_2) \leq 4n^4$. 
    We remark that in the next iteration, the pattern $P$ is in general different from the pattern in this iteration, 
    but that is not a problem, since $H$ contains a Hamiltonian cycle for all patterns.
\end{proof}

\mainTheoremOnlyIf*
\begin{proof}
    We give a short overview of the proof: Let $P \subseteq E'$ be an arbitrary pattern. We define depending on $P$ a perfect matching $M_P$ of $G_H$ and consider it together with the default matching $M_\text{def}$.
    By assumption $\diam(P_{G_H}) \leq 4n^4 + 46n$, so in particular $\dist(M_P, M_\text{def}) \leq 4n^4 + 46n$.
    We then proceed to show that in any sufficiently short flip sequence $(C_1, \dots, C_d)$ from $M_P$ to $M_\text{def}$ with $d \leq 4n^4 + 46n$, 
    for some $i \in [d]$ a cycle $C_i$ can be transformed into a Hamiltonian cycle respecting $P$ in $H$.
    Roughly speaking, the fact that $C_i$ mimics a Hamiltonian path will follow from the fact that most of the cycles $C_1,\dots,C_d$ must be regular (i.e.\ visit every city).
    Similarly, the fact that $C_i$ mimics a cycle respecting the pattern $P$ will follow from the fact that a regular cycle must (usually) visit a top-open ladder from the top and a bottom-open ladder from the bottom.
    Hence the fact $\diam(P_{G_H}) \leq 4n^4 + 46n$ implies that $H$ contains a Hamiltonian cycle respecting $P$. Since $P$ was arbitrary, this suffices to show.

    Let now $P \subseteq E'$ be a pattern. Let $A_1,\dots, A_k$ be the $\forall$-gadgets of $G_H$.
    We define a PM $M_P$ of $G_H$ as follows. 
    \begin{itemize}
        \item $M_P$ is in semi-default position.
        \item All tower gadgets are in locked position. 
        \item For all $j \in [k]$, we distinguish whether $e_j \in P$ or $\overline e_j \in P$. 
        If $e_j \in P$, then all the $n^4$ ladders of $A_j$ are top-open ladders.
        Otherwise, if $\overline e_j \in P$, then all the $n^4$ ladders of $A_j$ are bottom-open ladders.
    \end{itemize}
    We verify that $M_P$ is a well-defined perfect matching of $G_H$. 
    Since $M_P$ is in semi-default position, all vertices outside tower and ladder gadgets are touched exactly once.
    All vertices inside some tower or ladder gadget are also touched exactly once.

    Consider the PM $M_\text{def}$ as defined in \cref{sec:combining-the-pieces}. 
    By assumption there exists a flip sequence $\mathcal{C} = (C_1, \dots, C_d)$ with $d \leq 4n^4 + 46n$ from $M_P$ to $M_\text{def}$.
    By \cref{lem:regular-cycles} at most $1000n^2$ of these cycles are not regular, and $\mathcal{C}$ contains at least $4n^4 - 1002n^2$ regular cycles.
    For some $j \in [k]$ consider some fixed $\forall$-gadget $A_j$. The gadget $A_j$ contains $n^4$ ladders $L_1, \dots, L_{n^4}$. 
    Let us call a ladder \emph{corrupted}, if it is visited by at least one irregular cycle in $\mathcal{C}$. Consider the set 
    \[ 
    \La_j = \set{L_i : i \in [n^4], L_i \text{ is not corrupted }}
    \]
    of all uncorrupted ladders in $A_j$. 
    Due to \cref{lem:damage-irregular-cycle} we have $|\La_j| \geq n^4 - 4000n^2$. 
    Every ladder in $\La_j$ gets only visited by regular cycles. On the other hand, every regular cycle visits only one ladder of the gadget $A_j$ due to \cref{lem:forall-gagdet}.
    Since we have $d \leq 4n^4 +46n$, an average ladder from $\La_j$ gets visited at most
    \[
    \frac{4n^4 + 46n}{n^4 - 4000n^2}
    \]
    times, i.e.\ roughly 4 times. More precisely, let 
    \[ 
    \La_j' := \set{L \in \La_j : L \text{ gets visited exactly 4 times}}.
    \]
    Then we have 
    \[
    4|\La_j'| + 5|\La_j \setminus \La_j'| \leq d \leq 4n^4 + 46n.
    \]
    Since of course $|\La_j'| + |\La_j \setminus \La_j'| = |\La_j|$, this implies
    \begin{align*}
    |\La_j \setminus \La_j'| \leq 4n^4 + 46n - 4|\La_j| &\leq 16\,000n^2 + 46n \\
    \text{and } \quad \quad \quad \quad |\La_j'| = |\La_j| - |\La_j \setminus \La_j'| &\geq n^4 - 20\, 000n^2 - 46n.
    \end{align*}
    Now consider again the gadget $A_j$ and how it interacts with some regular cycle $C$. Due to \cref{lem:forall-gagdet}, 
    the cycle $C$ is either in top state or bottom state with respect to gadget $A_j$ and visits only one ladder from the same direction.
    Let us say $C$ respects the pattern $P$ at $A_j$, if either both $C$ is in top state and $e_j \in P$, or both $C$ is in bottom state and $\overline e_j \in P$. Let 
    \[
    \C_j := \set{C_i : i \in [d], C_i \text{ is regular and } C_i \text{ respects $P$ at $A_j$}}.
    \]
    Recall that we have defined $M_P$ to contain only top-open ladders inside $A_j$ if $e_j \in P$, and only bottom-open ladders inside $A_j$ if $\overline e_j \in P$.
    Note that every ladder in $\La'_j$ is visited only by regular (and therefore well-behaved) cycles, and exactly four of them. Hence the assumptions of \cref{lem:ladder-necessary} are met. Therefore these 4 cycles all need to come from the top if $L$ is top-open, 
    and they all need to come from the bottom if $L$ is bottom-open. 
    This shows
    \[
    |\C_j| \geq 4|\La_j'| \geq 4n^4 - 81\, 000n^2.
    \]
    This implies that for all $j \in [k]$ 
    \[
    |\C \setminus \C_j| \leq  d - 4n^4 + 81\, 000n^2 \leq 81\,000n^2 + 46n.
    \]
    Finally, using $|\C| \geq 4n^4 - 1002n^2$ we arrive at
    \[
    \bigcap_{j=1}^k \C_j \geq |\C| - \sum_{j=1}^k|\C \setminus \C_j| \geq 4n^4 - 83\, 000n^3.
    \]
    In particular, for $n  = |V(H)|$ large enough, there exists some cycle $C \in \bigcap_{j=1}^k \C_j$.
    This cycle $C \subseteq G_H$ is regular, i.e.\ it traverses every city gadget corresponding to some $x \in V(H)$ from $x_\vin$ to $x_\vout$. Furthermore, it traverses every $\forall$-gadget in the same manner as the pattern $P$.
    Therefore the directed graph $H$ has a Hamiltonian cycle respecting the pattern $P$.
\end{proof}

\section{Inapproximability}
\label{sec:inapprox}

In this section, we prove that for some $\varepsilon > 0$ the diameter (and the circuit diameter) of the bipartite perfect matching polytope 
cannot be approximated better than $(1 + \varepsilon)$ unless P = NP. 
We remark that our techniques inherently can show inapproximability only for a small $\varepsilon > 0$ (concretely we obtain $\varepsilon = 1/16\,226 \approx 6.1\cdot 10^{-5}$), and we do not optimize our reductions to obtain the best possible $\varepsilon$. To state our theorem completely formally, let us say an algorithm $\mathcal{A}$ is an approximation algorithm for the diameter with approximation guarantee $(1 + \varepsilon)$,
if $\mathcal{A}$ never underestimates the diameter, and overestimates it by at most a factor $(1 + \varepsilon)$. 
 This means
\[
\diam(P) \leq \mathcal{A}(P) \leq (1 + \varepsilon)\diam(P).
\]
(The same kind of approximation is considered in \cite{DBLP:conf/focs/Sanita18}). The main result of this section is then stated as follows.
\thmInapprox*

We remark that alternatively one could also consider approximation algorithms that are allowed to both over- 
and underestimate the true diameter, i.e.\ algorithms with 
\[
\max \set{\frac{\mathcal{A}(P)}{\diam(P)}, \frac{\diam(P)}{\mathcal{A}(P)}} \leq (1 + \varepsilon').\]
Since the proof of \cref{thm:inapproximability} is in terms of a gap-reduction, 
it follows that such two-sided approximation algorithms cannot have an approximation guarantee better than $(1 + \varepsilon')$, where $(1 + \varepsilon')^2 = (1 + \varepsilon)$, i.e.\ $\varepsilon' \approx \varepsilon/2$.

In the remainder of this section, the term 3SAT denotes the satisfiability problem, where the input is a boolean formula $\varphi$ in conjunctive normal form (CNF) such that every clause has at most 3 literals.
Max 3SAT is the maximization problem corresponding to 3SAT, where the goal is to find an assignment that satisfies as many of the clauses of $\varphi$ as possible.
The problem Max 3SAT-$d$ is equal to problem Max 3SAT with the additional restriction that for each variable $x_i$, the number of clauses that contain the literal $x_i$ or the literal $\overline x_i$ is at most $d$. 

The following is a direct consequence of the PCP theorem \cite{arora1998proof} in combination with constructions of certain expander graphs and is proven e.g.\ in Arora \cite[p.84]{arora1994probabilistic}, see also the survey of Trevisan \cite{DBLP:journals/eccc/ECCC-TR04-065}.

\begin{theorem}[Theorem 7 in \cite{DBLP:journals/eccc/ECCC-TR04-065}]
\label{thm:bounded-occurence}
    There are constants $d$ and $\varepsilon_1$ and a polynomial time computable reduction from 3SAT
    to Max 3SAT-$d$ such that if $\varphi$ is satisfiable then $f(\varphi)$ is satisfiable, and if $\varphi$ is not satisfiable then
    the optimum of $f(\varphi)$ is less than $1 - \varepsilon_1$ times the number of clauses.
\end{theorem} 

Using a standard strategy for inapproximability proofs, we can adapt \cref{thm:bounded-occurence} to our needs.
Specifically, consider a graph $G = (V, E)$. 
A \emph{closed walk} in $G$ is a sequence $W = (v_1, \dots, v_k)$ of vertices, such that $v_iv_{i+1} \in E$ for all $i \in [k-1]$ and also $v_kv_1 \in E$.
In contrast to a cycle, a closed walk may visit some vertices more than once.
Given a closed walk $W$ in $G$ and some $i \in \N_0$, we let
\[
W_i := \set{v \in V : W \text{ visits }v\text{ exactly $i$ times}}.
\]
\begin{definition}
\label{def:eps-good-cycle}
    Let $\varepsilon > 0$. A walk $W$ in a graph $G = (V, E)$ on $n = |V|$ vertices is called $\varepsilon$-good, if $|W_1| \geq (1 - \varepsilon)n$ (or equivalently $|V \setminus W_1| \leq \varepsilon n$).
\end{definition}

It seems likely that the following result is already known, however we were unable to locate a reference.
It proves hardness of approximation for walks that are \enquote{almost Hamiltonian}, i.e.\ $\varepsilon$-good.

\begin{theorem}
\label{thm:eps-good-reduction}
    There is a constant $\varepsilon_2 > 0$ and a polynomial time computable reduction that takes a 3SAT formula $\varphi$ and outputs an undirected graph $H$ such that
    \begin{itemize}
        \item If $\varphi$ is satisfiable, then $H$ contains a Hamiltonian cycle
        \item If $\varphi$ is not satisfiable, then $H$ does not contain an $\varepsilon_2$-good walk.
    \end{itemize}
\end{theorem} 
In order to prove \cref{thm:eps-good-reduction}, we rely on a folklore reduction from 3SAT to the Hamiltonian cycle problem. 
We show that this reduction without any modification satisfies also the stronger requirements of \cref{thm:eps-good-reduction}. The reduction is defined as follows:
Assume we are given a 3SAT instance $\varphi'$ in CNF with clauses $C_1,\dots,C_m$ and variables $x_1,\dots,x_k$. 

%------------------------------------------------------------------------------------
%------------------------------------------------------------------------------------
%------------------------------------------------------------------------------------
\tikzstyle{vertex}=[draw,circle,fill=black, minimum size=4pt,inner sep=0pt]
\tikzstyle{edge} = [draw,-]
\tikzset{
cross/.style={path picture={ 
  \draw[black]
(path picture bounding box.south east) -- (path picture bounding box.north west) (path picture bounding box.south west) -- (path picture bounding box.north east);
}}
}
\tikzstyle{XORGadget}=[draw, circle, cross, fill=white]
\tikzstyle{XOREdge}=[edge,rounded corners,{Stealth}-{Stealth}]
%------------------------------------------------------------------------------------
%------------------------------------------------------------------------------------
%------------------------------------------------------------------------------------
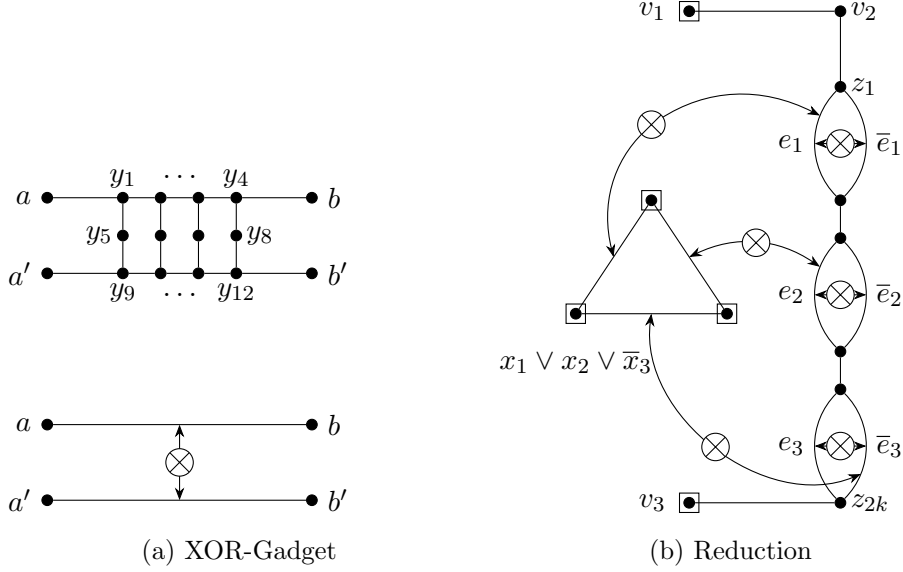
\begin{figure}
\centering
\begin{subfigure}{0.4\textwidth}
\begin{tikzpicture}
\node[vertex,label=left:$a$] (x1) at (-0.5,1) {};
\node[vertex] (x2) at (0.5,1){};
\node[vertex] (x3) at (1,1){};
\node[vertex] (x4) at (1.5,1){};
\node[vertex] (x5) at (2,1){};
\node[vertex,label=right:$b$] (x6) at (3,1){};
\node[vertex] (y1) at (0.5,0.5){};
\node[vertex] (y2) at (1,0.5){};
\node[vertex] (y3) at (1.5,0.5){};
\node[vertex] (y4) at (2,0.5){};
\node[vertex,label=left:$a'$] (z1) at (-0.5,0){};
\node[vertex] (z2) at (0.5,0){};
\node[vertex] (z3) at (1,0){};
\node[vertex] (z4) at (1.5,0){};
\node[vertex] (z5) at (2,0){};
\node[vertex,label=right:$b'$] (z6) at (3,0){};
\node[above] at (x2) {$y_1$};
\node[above=8, right=-3] at (x3) {$\dots$};
\node[above] at (x5) {$y_4$};
\node[left] at (y1) {$y_5$};
\node[right] at (y4) {$y_8$};
\node[below] at (z2) {$y_{9}$};
\node[below=8, right=-3] at (z3) {$\dots$};
\node[below] at (z5) {$y_{12}$};
\draw[edge] (x1) -- (x2) -- (x3) -- (x4) -- (x5) -- (x6);
\draw[edge] (z1) -- (z2) -- (z3) -- (z4) -- (z5) -- (z6);
\draw[edge] (x2) -- (y1) -- (z2);
\draw[edge] (x3) -- (y2) -- (z3);
\draw[edge] (x4) -- (y3) -- (z4);
\draw[edge] (x5) -- (y4) -- (z5);
\node[vertex,label=left:$a$] (v1) at (-0.5,-2) {};
\node[vertex,label=right:$b$] (v2) at (3,-2){};
\node[vertex,label=left:$a'$] (w1) at (-0.5,-3){};
\node[vertex,label=right:$b'$] (w2) at (3,-3){};
\draw[edge] (v1) to (v2);
\draw[edge] (w1) to (w2);
\draw[XOREdge] ($(v1)!0.5!(v2)$) to coordinate (cross1) ($(w1)!0.5!(w2)$);
\node[XORGadget] at (cross1) {}; 
\end{tikzpicture}
\caption{XOR-Gadget}
\label{fig:hamilton-classic-reduction:xor}
\end{subfigure}
%------------------------------------------------------------------------
\begin{subfigure}{0.4\textwidth}
\begin{tikzpicture}
\node[vertex] (s1) at (2,6.5) {};
\node[left=5] at (s1) {$v_1$};
\node[vertex] (s2) at (4,6.5) {};
\node[right] at (s2) {$v_2$};
\node[vertex] (x1) at (4,5.5) {};
\node[right] at (x1) {$z_1$};
\node[vertex] (x1') at (4,4) {};
\node[vertex] (x2) at (4,3.5) {};
\node[vertex] (x2') at (4,2) {};
\node[vertex] (x3) at (4,1.5) {};
\node[vertex] (x3') at (4,0) {};
\node[right] at (x3') {$z_{2k}$};
\node[vertex] (s4) at (2,0) {};
\node[left=5] at (s4) {$v_3$};
\node[draw] at (s1) {};
\node[draw] at (s4) {};
\draw[edge, bend right=45] (x1) to coordinate (a1) coordinate[pos=0.25] (a1') coordinate[pos=0.5] (q1') (x1');
\node[left] at (a1) {$e_1$}; 
\draw[edge, bend left=45] (x1) to  coordinate[pos=0.5] (q1) node[right]{$\overline e_1$} (x1');
\draw[edge, bend right=45] (x2) to coordinate (a2) coordinate[pos=0.25] (a2') coordinate[pos=0.5] (q2') (x2');
\node[left] at (a2) {$e_2$};
\draw[edge, bend left=45] (x2) to coordinate[pos=0.5] (q2) node[right]{$\overline e_2$} (x2');
\draw[edge, bend left=45] (x3) to coordinate[pos=0.5] (q3) coordinate (a3) coordinate[pos=0.75] (a3') (x3');
\node[right] at (a3) {$\overline e_3$};
\draw[edge, bend right=45] (x3) to coordinate[pos=0.5] (q3') node[left]{$e_3$} (x3');
\draw[edge] (s1) -- (s2) -- (x1);
\draw[edge] (x3') to (s4);
\node[vertex] (c1) at (1.5,4) {};
\node[vertex] (c2) at (0.5,2.5) {};
\node[vertex] (c3) at (2.5,2.5) {};
\draw[edge] (c1) -- (c2) -- (c3) -- (c1);
\draw[edge] (x1') -- (x2);
\draw[edge] (x2') -- (x3);
\node[draw] at (c1) {};
\node[draw] at (c2) {};
\node[draw] at (c3) {};
\coordinate (cross1) at (1.5,5){}; 
\draw[XOREdge, bend right] (a1') to (cross1) to ($(c1)!0.5!(c2)$);
\node[XORGadget]  at (cross1) {}; 
\draw[XOREdge, bend right] (a2') to coordinate (cross2) 
($(c1)!0.5!(c3)$);
\node[XORGadget] at (cross2) {}; 
\draw[XOREdge, bend left=60] (a3') to coordinate (cross3) ($(c2)!0.5!(c3)$);
\node[XORGadget] at (cross3) {}; 
\node[below=10] at (c2) {$x_1 \lor x_2 \lor \overline x_3$};
\draw[XOREdge] (q1) to coordinate[pos=0.5] (cross4) (q1'); 
\node[XORGadget] at (cross4) {}; 
\draw[XOREdge] (q2) to coordinate[pos=0.5] (cross5) (q2'); 
\node[XORGadget] at (cross5) {}; 
\draw[XOREdge] (q3) to coordinate[pos=0.5] (cross6) (q3'); 
\node[XORGadget] at (cross6) {}; 
\end{tikzpicture}
\caption{Reduction}
\label{fig:hamilton-classic-reduction:reduction}
\end{subfigure}
\caption{Folklore reduction of 3SAT to the Hamiltonian cycle problem. Arrows marked with a cross denote an XOR-gadget, as illustrated in (a). Vertices marked with a square are all connected in one clique.}
\label{fig:hamilton-classic-reduction}
\end{figure}
%------------------------------------------------------------------------------------
%------------------------------------------------------------------------------------
%------------------------------------------------------------------------------------
The reduction is depicted in \cref{fig:hamilton-classic-reduction}.
Let $ab$ and $a'b'$ be two edges. We first define XOR-gadgets between $ab$ and $a'b'$ by introducing additional vertices $y_1, \dots, y_{12}$ and connecting them like in \cref{fig:hamilton-classic-reduction}.
They are similar to the XOR-gagdets in \cref{sec:towers-and-cities}, with the exception that we replace the city gadgets with single vertices.
Analogously to before, we can conclude that every Hamiltonian cycle uses exactly one of the two edges $ab$ and $a'b'$ of a XOR-gadget.
Given the 3SAT formula $\varphi$, we can then define an undirected graph $H$ as follows: $H$ contains $k$ variable gadgets, $m$ clause gadgets, and three additional vertices $v_1,v_2,v_3$ as depicted in \cref{fig:hamilton-classic-reduction}.
Here, variable gadgets are defined as follows: The graph $H$ includes vertices $z_1, \dots, z_{2k}$. For all $i \in [k]$ there are two parallel edges from $z_{2i-1}$ to $z_{2i}$, which we denote by $e_i$ and $\overline e_i$.
We interpret them as the edges corresponding to literal $x_i$ and literal $\overline x_i$.
We connect $e_i$ and $\overline e_i$ with a XOR-gadget.
Note that after the XOR-gadget is applied, there is no multi-edge anymore, i.e.\ we still have that $H$ is a simple graph.
A clause gadget is defined as follows: For some clause $C_j$ with $C_j = \ell_1 \lor \dots \lor \ell_{t_j}$ where $t_j \in \set{1,2,3}$, 
we introduce $t_j$ vertices $u^{(j)}_1,\dots, u^{(j)}_{t_j}$ in $H$, and we connect them in a cycle of length $t_j$ (i.e.\ a loop if $t_j = 1$, two parallel edges if $t_j=2$, and a triangle if $t_j=3$).
For $i = 1, \dots, t_j$, a XOR-gadget connects the $i$-th edge of this cycle to the edge inside the variable gadget corresponding to the literal $\ell_i$. 
Note that again this makes the graph simple.
We add the edges $v_1v_2, v_2z_1, z_{2k}v_3$ and $z_{2i}z_{2i+1}$ for all $i=1,\dots,k-1$.
Finally, in the vertex set
\[
V' := \bigcup_{j=1}^m \set{u^{(j)}_1,\dots,u^{(j)}_{t_j}} \cup \set{v_1, v_3},
\]
every vertex is connected to each other (i.e.\ $H$ is a clique induced on $V'$).
The set $V'$ is highlighted with square marks in \cref{fig:hamilton-classic-reduction}.

We quickly verify that this folklore reduction is indeed a correct reduction from 3SAT to Hamiltonian cycle.
Indeed, every Hamiltonian cycle $F$ visits $v_2$ (note that $v_2 \not\in V'$). Hence every Hamiltonian cycle traverses the variable gadgets in order and 
uses exactly one of the two edges $e_i$ or $\overline e_i$ for $i \in [k]$.
(For simplicity of notation, we say that the cycle $F$ uses $e_i$ even though due to the XOR-gadgets the edge $e_i$ is not an actual edge of $H$.)
If for some clause gadget corresponding to clause $C_j$ the cycle $F$ makes the wrong variable choices, i.e.\ includes none of the literals of $C_j$, 
then due to the XOR-gadgets, the cycle $F$ restricted to the clause gadget of $C_j$ is a already a cycle.
But that is a contradiction, since a Hamiltonian cycle cannot contain a smaller cycle.
Hence we conclude that if $F$ is a Hamiltonian cycle in $H$, then $\varphi'$ is satisfiable.
On the other hand, due to $V'$ inducing a clique, if $\varphi$ is satisfiable, then $H$ has a Hamiltonian cycle.

We now show that the same reduction also satisfies the requirements of \cref{thm:eps-good-reduction}.
The main idea behind the proof is the following: 
Since every $\varepsilon$-good walk $W$ has $|W_1| \geq (1 - \varepsilon)|V(H)|$, we have that the majority of all gadgets in $H$ are entirely contained in $W_1$, 
i.e.\ all their vertices are visited exactly once by $W$.
We show that this implies that most of the gadgets work as intended, 
hence for $\varepsilon_2$ small enough, a $\varepsilon_2$-good walk $W$ implies that a fraction of at least $(1 - \varepsilon_1)$ clauses of $\varphi$ can be made true.

\begin{proof}[Proof of \cref{thm:eps-good-reduction}]
    Assume we are given some 3SAT instance $\varphi$.
    Let $\varepsilon_1, d$ be constants as in \cref{thm:bounded-occurence}.
    In a first step we transform $\varphi$ into an instance $\varphi' = f(\varphi)$ of Max 3SAT-$d$ such that $\varphi'$ has the properties described in \cref{thm:bounded-occurence}. In a second step, we transform $\varphi'$ into a graph $H = (V,E)$ as described above.
    We let
    \[
    \varepsilon_2 := \frac{ \varepsilon_1 }{61(d+1)}. 
    \]
    If $\varphi'$ is satisfiable, then $H$ contains a Hamiltonian cycle, as proven above. 
    So it remains to prove that if $\varphi'$ is not satisfiable, then $H$ does not contain a $\varepsilon_2$-good walk.
    For the sake of contradiction, assume $H$ contains a $\varepsilon_2$-good walk $W$, we show that $\varphi'$ is satisfiable.
    Let $n := |V(H)|$. Let $m$ be the number of clauses of $\varphi'$. Note that the graph $H$ contains $k$ variable gadgets, $m$ clause gadgets, at most $3m + k$ XOR-gadgets, and the three additional vertices $v_1,v_2,v_3$. 
    Since these in total are all the vertices of  $H$, we have (using $k \leq 3m$ and w.l.o.g.\ $m \geq 4$)
    \[
    n \leq 2k + 3m + k + 16(3m) + 3 \leq 60m + 3 < 61m.
    \]
    Consider some XOR-gadget $g$ in the graph $H$. We say that the XOR-gadget $g$ consists of precisely the 16 vertices 
    $V_g := \set{a,b,a',b'} \cup \fromto{y_1}{y_{12}}$ as shown in  \cref{fig:hamilton-classic-reduction}.
    Let $i \in [k]$ and consider the 3SAT variable $x_i$. 
    Observe that the variable gadget of $x_i$ contains the XOR-gadget $g_0$ between $e_i, \overline e_i$ and in addition to that it is connected to at most $d$ XOR-gadgets $g_1,\dots,g_t$ with $t \leq d$.
    Let us say that with respect to the walk $W$, the variable $x_i$ is \emph{uncorrupted}, if 
    \[
        \left( \set{z_{2i-1}, z_{2i}} \cup \bigcup_{j=0}^t V_{g_j} \right) \subseteq W_1,
    \]
    and it is corrupted otherwise. Since the above vertex sets are pairwise disjoint for all $i \in [k]$, we have that every vertex in $V \setminus W_1$ can corrupt at most one variable. Hence
    \[
    |\set{i \in [k] : x_i \text{ is corrupted}}| \leq |V \setminus W_1| \leq \varepsilon_2 n.
    \]
    We claim that if a variable $x_i$ is uncorrupted, then the walk $W$ traverses the variable gadget of $x_i$ either
    only using edge $e_i$, or only using edge $\overline e_i$.
    Indeed, consider the walk $W$ and its interaction with some XOR-gadget $g$.
    Since the XOR-gadget is contained entirely in $W_1$ (i.e.\  $V_g \subseteq W_1$), we can see that the walk $W$ enters and leaves every of the vertices $y_5,\dots,y_8$ exactly once.
    Furthermore, for each $y \in \set{y_5,\dots,y_8}$, the walk enters and leaves $y$ using two different edges, since otherwise the walk would use a neighbor of $y$ more than once (but all neighbors of $y$ belong to $W_1$).
    Hence the walk $W$ contains the subwalks $(y_i,y_{i+4},y_{i+8})$ for all $i=1,\dots,4$.
    By extending this argument to $y_1,\dots,y_4$ and $y_9,\dots,y_{12}$, we see that $W$ uses exactly one of the two edges $ab$ and $a'b'$.
    Now, consider vertex $z_{2i-1}$. 
    Since an uncorrupted XOR-gadget connects $e_i$ and $\overline e_i$, the walk has to use one of the two edges $e_i, \overline e_i$.
    On the other hand, if the walk would use both edges $e_i, \overline e_i$, then because all XOR-gadgets (attached to either $e_i$ or $\overline e_i$) work properly,
    the walk meets itself at $z_{2i}$. Hence $z_{2i} \not\in W_1$, a contradiction. This proves the claim.
    
    For $j \in [m]$ consider the clause $C_j$. It contains at most three literals, i.e.\ $C_j = \ell_1 \lor \dots \lor \ell_{t_j}$ with $t_j \in \set{1,2,3}$.
    Let us call the clause $C_j$ corrupted, if one of the variables corresponding to its literals  $\ell_1,\dots,\ell_{t_j}$ is corrupted, 
    or if one of the $t_j$ vertices $u^{(j)}_1,\dots, u^{(j)}_{t_j}$ of the clause gadget is not contained in $W_1$. 
    For the first reason at most $d \varepsilon_2 n$ clauses are corrupted, because a corrupted variable corrupts at most $d$ clauses.
    For the second reason, at most $|V \setminus W_1| \leq \varepsilon_2 n$ clauses are corrupted, since the corresponding vertex sets are disjoint. Therefore
    \[
    |\set{j \in [m] : C_j \text{ is corrupted}}| \leq d\varepsilon_2 n + \varepsilon_2 n < 61(d+1)\varepsilon_2 m = \varepsilon_1 m.
    \]
    For $i \in [k]$, we can now consider the natural variable assignment
    \[
    \alpha(x_i) := \begin{cases} 
            1 &\text{if $x_i$ uncorrupted, and $W$ uses $e_i$}\\
            0 &\text{if $x_i$ uncorrupted, and $W$ uses $\overline e_i$}\\
            \text{arbitrary} &\text{if $x_i$ corrupted.} 
        \end{cases}
    \]
    For every uncorrupted clause $C_j$ for some $j \in [m]$, we claim that the assignment $\alpha$ satisfies $C_j$. Indeed, note that all of its variables are uncorrupted, 
    and so all of the XOR-gadgets attached to $C_j$ work as intended.
    Therefore if $\alpha$ does not satisfy $C_j$, then the walk $W$ describes a cycle restricted to the clause gadget, and so one of the vertices $u^{(j)}_1,\dots, u^{(j)}_{t_j}$ is visited at least twice by $W$, a contradiction.
    Since there are less than $\varepsilon_1 m$ corrupted clauses, due to \cref{thm:bounded-occurence}, formula $\varphi'$ has even an assignment that satisfies every clause.
    This was to show.
\end{proof}

\thmInapprox*

\begin{proof}
We show that there exists a polynomial-time computable reduction, that given a 3SAT formula $\varphi$ returns a bipartite graph $G$, such that for some constant $n_0 \in \N$
\begin{itemize}
    \item If $\varphi$ is satisfiable, then $\diam(P_{G}) \leq 2n^2 + 2n$
    \item If $\varphi$ is not satisfiable, then $\diam(P_G) > (1 + \varepsilon)(2n^2 + 2n)$ for all $n \geq n_0$.
\end{itemize}
This suffices to prove the theorem.
Given a SAT formula $\varphi$, due to \cref{thm:eps-good-reduction} there exists a constant $\varepsilon_2 > 0$ and a polynomial-time computable graph $H$ such that if $\varphi$ is satisfiable, 
$H$ has a Hamiltonian cycle, and if $\varphi$ is not satisfiable, $H$ does not have an $\varepsilon_2$-good walk.
We choose some constant $0 < \delta < \varepsilon_2$ very close to zero and let
\[
\varepsilon := \frac{1}{(1 + \delta)(1 - \varepsilon_2)} - 1.
\]
Note that $\varepsilon > 0$, since $0 \leq (1 + \delta)(1 - \varepsilon_2) = 1 + \delta - \varepsilon_2 - \delta \varepsilon_2 < 1$.
We construct a graph $G$ from $H$ as follows. We let $n := |V(H)|$. First we consider the vertex set 
\[
\bigcup_{v \in V(H)} \set{v_\vin, v_\vout}.
\]
Then we connect for all $v \in V(H)$ the two vertices $v_\vin, v_\vout$ with a city gadget of height $n^2$ and width $4n^2 + 4n$.
Finally, we add the edge set
\[
\bigcup_{ab \in E(H)} \set{a_\vout b_\vin, b_\vout a_\vin} 
\]
to $H$. 
We remark that this construction of $G$ resembles the original construction of \cite{nobel2025complexity}.
It is similar to our construction of \cref{sec:combining-the-pieces}, with the difference that $H$ is undirected 
(in \cref{sec:combining-the-pieces} for technical reasons it is more convenient to work with a directed graph, while in the current section it is more convenient to work with an undirected graph).
The graph $G$ is bipartite, since city gadgets connect their endpoints via an odd-length path, and all other other edges connect some vertex $a_\vout$ to some vertex $b_\vin$.
The proof of \cref{thm:inapproximability} is now completed with the following two \cref{lem:inapproximability-if,lem:inapproximability-only-if}.
\end{proof}

\begin{lemma}
    \label{lem:inapproximability-if}
    If $\varphi$ is satisfiable, then $\diam(P_{G}) \leq 2n^2 + 2n$.
\end{lemma}
\begin{proof}
    Let $M_1, M_2$ be arbitrary PMs of $G$, we will show $\dist(M_1, M_2) \leq 2n^2 + 2n$. 
    Since $M_1,M_2$ are arbitrary, this suffices to show.
    Call some PM $M$ in \emph{semi-default state}, if every city gadget is in semi-default state with respect to $M$.
    Observe that $G$ contains certainly at least one PM $M_\text{def}$ which is in semi-default state.
    Consider the vertex set $V_s := \set{v_\vin : v \in V(H)}$. We have $|V_s| = n$.
    Furthermore, a PM $M$ is in semi-default state if and only if the PM $M$ assigns the same partners to vertices of $V_s$ as the matching $M_\text{def}$.
    Analogously to \cref{lem:semi-default-dense} we conclude that there exist semi-default PMs $M'_1, M'_2$ such that $\dist(M_1, M'_1) \leq n$ and $\dist(M_2, M'_2) \leq n$.
    Let $T$ be fixed tower gagdet in $G$. 
    Since $T$ has height $n^2$, due to \cref{lem:tower:upper-bound} there is a well-behaved flip sequence $(P^{(T)}_1,\dots, P^{(T)}_d)$ of length $d = 2n^2$ that transforms $M'_1$ into $M'_2$ restricted to $T$.
    Since $\varphi$ is satisfiable, $G$ contains a Hamiltonian cycle $F$. 
    We can consider for $i = 1,\dots, 2n^2$ the cycle $C_i$ defined as follows:
    $C_i$ visits every city gadget.
    Inside some tower $T$, the cycle $C_i$ is equal to $P^{(T)}_i$.
    Outside the towers, it follows globally the same route as the Hamiltonian cycle $F$, i.e.\ it uses all the edges $\set{a_\vout b_\vin : ab \in E(F)}$.
    Analogously to \cref{lem:main-thm-if}, we see that for each $i \in [2n^2]$, the object $C_i$ is a well-defined cycle. 
    Since $M'_1$ is in semi-default state, one can show by induction that the cycle $C_i$ is alternating with respect to $M'_1 \symdiff C_1 \symdiff \dots \symdiff C_{i-1}$ for all $i \in [d]$.
    Furthermore, the flip sequence $(C_1,\dots, C_d)$ transforms $M'_1$ into $M'_2$. 
    We conclude in total
    \[
    \dist(M_1, M_2) \leq \dist(M_1,M'_1) + \dist(M'_1, M'_2) + \dist(M'_2, M_2) \leq n + 2n^2 + n.
    \]

\end{proof}

\begin{lemma}
    \label{lem:inapproximability-only-if}
     If $\varphi$ is not satisfiable, then for all $n$ larger than some constant $n_0$, we have $\diam(P_{G}) > (1 + \varepsilon)(2n^2 + 2n)$.
\end{lemma}
\begin{proof}
    We consider the contrapositive, that is we that assume $\diam(P_G) \leq (1 + \varepsilon)(2n^2 + 2n)$, and we want show that $\varphi$ is satisfiable.
    Consider PMs $M_1, M_2$ of $G$ defined as follows. 
    Both $M_1, M_2$ have all city gadgets in semi-default position.
    Furthermore, $M_1$ has all tower gadgets in locked position (see \cref{fig:tower-gadget:locked}).
    In contrast, $M_2$ has all tower gadgets in default position (see \cref{fig:tower-gadget:default}).
    Note that $M_1$ and $M_2$ are well-defined PMs of $G$.
    Let $d := \dist(M_1, M_2)$. 
    By assumption on the diameter, we have $d \leq (1 + \varepsilon)(2n^2 + 2n)$ and there exists a flip sequence $(C_1,\dots, C_d)$ from $M_1$ to $M_2$.
    Consider for some vertex $v \in V(H)$ the corresponding city gadget $g_v$. By definition its width is $t = 4n^2 + 4n$. 
    Due to $\varepsilon < 1$, we have $d < t$. 
    This implies that among the $t$ towers in city $g_v$, there exists one tower $T^\star_v$ with the property that none of the cycles $C_1, \dots, C_d$ is entirely contained in $T^\star_v$.
    Hence every cycle of $(C_1,\dots,C_d)$ either is well-behaved for $T^\star_v$ or does not visit the city $g_v$ at all.
    For each $i \in [d]$, let $p_i$ be the number of different towers $T^\star_v$ that the cycle $C_i$ visits, formally
    \[
    p_i := |\set{v \in V(H) : C_i \text{ visits }T^\star_v}|.
    \]
    
    We apply \cref{lem:tower-lower-bound} to the tower $T^\star_v$, and see that at least $2n^2 - 2$ cycles need to visit $T^\star_v$ in order to transform $M_1$ into $M_2$ (restricted to $T^\star_v$).
    On the other hand, a fixed cycle $C_i$ can visit each tower $T^\star_v$ at most once.
    By double counting we obtain
    \[
    \sum_{i = 1}^d p_i = |\set{(i,v) \in [d] \times V(H) : C_i \text{ visits }T^\star_v}| \geq (2n^2 - 2)n.
    \]
    Since $\delta$ is a constant, for $n \geq n_0(\delta)$ large enough we have $(1 + \delta)(2n^3 - 2n) \geq (2n^3 + 2n^2)$. Then by the pigeonhole principle, there exists a $j \in [d]$ such that
    \[
    p_j \geq \frac{2n^3 - 2n}{(1 + \varepsilon)(2n^2 + 2n) } \geq \frac{2n^3 + 2n^2}{(1 + \varepsilon)(1 + \delta)(2n^2 + 2n)} = \frac{n}{(1 + \varepsilon)(1 + \delta)} = (1 - \varepsilon_2)n.
    \]
    Therefore the cycle $C_j$ visits at least $(1 - \varepsilon_2)n$ city gadgets.
    We claim that we can modify the cycle $C_j \subseteq G$ to obtain an $\varepsilon_2$-good walk $W$ in $H$.
    We start with $W = \emptyset$ and modify $W$ according to the following rule:
    \begin{itemize}
        \item Whenever $C_j$ visits some vertex $v_\vin$ in $G$ (where $v \in V(H)$), the next step of $C_j$ is either to traverse the city gadget and visit $v_\vout$, 
        or to traverse some edge of the form $v_\vin w_\vout$.
        In the first case, we do not modify $W$, in the second case we add $(v, w)$ to the walk $W$. 
        Note that this is legal since $wv = vw$ is an edge in $H$.
        \item Analogously, whenever $C_j$ visits some vertex $v_\vout$ in $G$ (where $v \in V(H)$), the next step of $C_j$ is either to traverse the city gadget and visit $v_\vin$, 
        or to traverse some edge of the form $v_\vout w_\vin$.
        In the first case, we do not modify $W$, in the second case we add $(v, w)$ to the walk $W$.
        Again this is legal since $vw$ is an edge in $H$.
    \end{itemize}
    Note that the above procedure indeed constructs a well-defined closed walk $W$ in $H$. (If $C_j$ is currently at one of the two vertices $v_\vin$ or $v_\vout$, then $W$ is currently at $v$).
    Furthermore, note that if $C_j$ visits a city gadget $g_v$ corresponding to some $v \in V(H)$, then $v \in W_1$. 
    This is because $C_j$ is a simple cycle and hence cannot visit a city twice.
    Since $p_j \geq (1 - \varepsilon_2)n$, we have $|W_1| \geq (1 - \varepsilon_2)n$, i.e.\ $W$ is $\varepsilon_2$-good.
    Since $H$ was constructed according to \cref{thm:eps-good-reduction}, we have that the formula $\varphi$ is satisfiable.
\end{proof}

\textbf{Concrete value of $\varepsilon$}.
Arora \cite{arora1994probabilistic} proves \cref{thm:bounded-occurence} for constants $d = 13$ and $\varepsilon_1 = 1/19$.
Hence we obtain $\varepsilon_2 = 1/16\,226 \approx 6.1\cdot 10^{-5}$ and finally $\varepsilon = \varepsilon_2$ (note that for $\delta$ arbitrarily small, $\varepsilon$ can be made arbitrarily close to $1/(16\,226 - 1)$, where as a trade-off $n_0$ gets bigger). 
Since we did not optimize our reductions for the best possible $\varepsilon$, this can likely be improved.

\section{Conclusion}
\label{sec:conclusion}
In this work, we proved that computing the diameter (the circuit diameter, respectively) of the bipartite perfect matching polytope is $\Pi^p_2$-complete. 
As a result, computation of the diameter can be compared to computation of other notoriously hard parameters like the generalized Ramsey number.
Furthermore, we showed that the diameter (the circuit diameter, respectively) cannot be approximated in polynomial time with a factor better than $1 + \varepsilon$, unless P = NP.
We remark that these two main results are incomparable.
On the one hand the second result is stronger, because it is an inapproximability result.
On the other hand, the second result is weaker, since it only shows NP-hardness instead of $\Pi^p_2$-hardness.

There are many open questions that remain.
It seems likely that our value of $\varepsilon \approx 6.1 \cdot 10^{-5}$ can be improved, especially when comparing against the related result of Cardinal and Steiner for polytope distances \cite{DBLP:conf/ipco/CardinalS23}.
We hence ask the question, whether the (circuit) diameter of a polytope does or does not admit a constant-factor approximation.
One could also ask, whether it is possible to combine our two main results, i.e.\ whether approximation of the (circuit) diameter is $\Pi^p_2$-hard.
We remark however that even though theoretically valid, results of this kind seem to be very sparse in the literature.
One reason for this might be that approximation algorithms are always required to be polynomial time, 
and so it could be argued that a completeness result for a higher class than NP offers less insights.
The author of this paper would be interested in such a result anyway.

Our $\Pi^p_2$-completeness result does not hold for the monotone diameter as defined in \cite{nobel2025complexity}, 
since it follows from the arguments presented there that the decision problem for the monotone diameter of the bipartite perfect matching polytope is contained in NP. 
We therefore ask if for some other polytope it can be proven that computation of the monotone diameter is $\Pi^p_2$-hard.
Furthermore, Kaibel and Pfetsch \cite{DBLP:conf/dagstuhl/KaibelP03} ask for the complexity of computing the diameter of a \emph{simple} polytope, 
i.e.\ a polytope $P \subseteq \R^d$ where every vertex has degree $d$. Since the bipartite perfect matching polytope is not simple, this question remains unanswered.

As a final question, we ask whether our $\Pi^p_2$-completeness and inapproximability results can also be extended to the case of bipartite, planar graphs of maximum degree 3, 
similar as in \cite{DBLP:journals/tcs/ItoDHPSUU11}. 
Our current \cref{lem:damage-irregular-cycle} provides a barrier for this result, since it states that the four vertices $x_2,x_3,x_6,x_7$ inside a $\forall$-gadget are essential to the proof. 
However, they have an arbitrarily high degree and together with their surroundings create a non-planar graph.
Hence in order to prove hardness in this restricted case, a different approach may be needed.

\bibliographystyle{alpha} 
\bibliography{references}

\end{document}